\pdfoutput=1
\documentclass[12pt,a4paper]{amsart}
\usepackage[a4paper,inner=2.5cm,outer=2.5cm,top=2.5cm,bottom=2.5cm]{geometry}
\usepackage{amsmath,amssymb,amsthm,enumerate,mathtools,stmaryrd
}
\usepackage{hyperref}
\hypersetup{colorlinks=true,linkcolor=blue,citecolor=teal,filecolor=magenta,urlcolor=cyan}

\usepackage[oldstyle]{libertine}%
\usepackage[T1]{fontenc}
\usepackage[utf8]{inputenc}
\makeatletter
\renewcommand*\libertine@figurestyle{LF}
\makeatother
\makeatletter
\renewcommand*\libertine@figurestyle{OsF}

\usepackage{makecell}
\usepackage{graphicx}

\usepackage{float}

\usepackage{extarrows}

\usepackage{xcolor}

\newcommand{\C}{\mathbb{C}}

\usepackage{url}

\usepackage[backend=bibtex,style=alphabetic,sorting=nyt,isbn=false,url=false,doi=true,maxalphanames=10,minalphanames=4,mincitenames=4,maxcitenames=10,minnames=4,maxnames=10,giveninits=true,maxbibnames=99]{biblatex}
\theoremstyle{plain}
\newtheorem{theorem}{Theorem}[section]
\newtheorem{proposition}[theorem]{Proposition}
\newtheorem{corollary}[theorem]{Corollary}
\newtheorem{lemma}[theorem]{Lemma}

\newtheorem{conjecture}[theorem]{Conjecture}
\theoremstyle{definition}
\newtheorem{definition}[theorem]{Definition}

\makeatletter
\@addtoreset{proofpart}{theorem}
\makeatother
\theoremstyle{remark}
\newtheorem{remark}[theorem]{Remark}

\newcommand{\cS}{\mathcal{S}}

\newcommand{\cP}{\mathcal{P}}

\DeclareMathOperator{\Aut}{Aut}

\newcommand{\res}{\mathop{\rm res}}

\newcommand{\coeff}[2]{\mathop{[{#1}^{#2}]}}
\newcommand{\restr}[2]{\mathop{\big\lfloor_{{#1}\to {#2}}}}
\newcommand{\set}[1]{\llbracket {#1} \rrbracket}

\title[Log topological recursion]{Log topological recursion through the prism of $x-y$ swap}

\author[A.~Alexandrov]{A.~Alexandrov}
\address{A.~A.: Center for Geometry and Physics, Institute for Basic Science (IBS), Pohang 37673, Korea
}
\email{alex@ibs.re.kr}

\author[B.~Bychkov]{B.~Bychkov}
\address{B.~B.: Department of Mathematics, University of Haifa, Mount Carmel, 3498838, Haifa, Israel}
\email{bbychkov@hse.ru}

\author[P.~Dunin-Barkowski]{P.~Dunin-Barkowski}
\address{P.~D.-B.: Faculty of Mathematics, HSE University, Usacheva 6, 119048 Moscow, Russia; HSE--Skoltech International Laboratory of Representation Theory and Mathematical Physics, Skoltech, Bolshoy Boulevard 30 bld. 1, 121205 Moscow, Russia; and NRC “Kurchatov Institute” -- ITEP, 117218 Moscow, Russia}
\email{ptdunin@hse.ru}

\author[M.~Kazarian]{M.~Kazarian}
\address{M.~K.: Faculty of Mathematics, HSE University, Usacheva 6, 119048 Moscow, Russia; and Igor Krichever Center for Advanced Studies, Skoltech, Bolshoy Boulevard 30 bld. 1, 121205 Moscow, Russia}
\email{kazarian@mccme.ru}

\author[S.~Shadrin]{S.~Shadrin}
\address{S.~S.: Korteweg-de Vries Institute for Mathematics, University of Amsterdam, Postbus 94248, 1090GE Amsterdam, The Netherlands}
\email{S.Shadrin@uva.nl}		

\begin{document}
	
\begin{abstract} We introduce a new concept of \emph{logarithmic topological recursion} that provides a patch to topological recursion in the presence of logarithmic singularities and prove that this new definition satisfies the universal $x-y$ swap relation. This result provides a vast generalization and a proof of a very recent conjecture of Hock. It also uniformly explains (and conceptually rectifies) an approach to the formulas for the $n$-point functions proposed by Hock.
\end{abstract}
	
\maketitle
	
\setcounter{tocdepth}{3}
\tableofcontents

\section{Introduction}

\subsection{Topological recursion} \label{sec:toprecintro}

Topological recursion~\cite{EO} associates to a small collection of input data that consists of two meromorphic functions $x$ and $y$ defined on a compact Riemann surface $\Sigma$ (the triple ($\Sigma,x,y$) is traditionally called the spectral curve) a system
\begin{align}
	\{\omega^{(g)}_n\}_{\substack{g\geq 0, n\geq 1\\ 2g-2+n>0}}
\end{align} of symmetric meromorphic $n$-differentials $\omega^{(g)}_n$ on $\Sigma^n$, $n\geq 1$, via an explicitly given recursive procedure. It is one of the principal tools on the edge between mathematical physics, algebraic geometry, enumerative combinatorics, and integrable systems, moreover, it simultaneously generalizes and to a large extent replaces a variety of techniques such as
\begin{itemize}
	\item matrix models techniques in mathematical physics, see e.g.~\cite{CEO,EO-1st,ZhouKW};
	\item localization formulas in algebraic geometry, see e.g.~\cite{EO-3folds,DKOSS-ELSV,DLPS,FLZ,DKPS-qr-ELSV};
	\item cut-and-join type operators in enumerative combinatorics, see e.g.~\cite{MulaseCaJ,DMSS,DoCaJ,DOPS,DKPS-cut-and-join};
	\item Dubrovin--Frobenius manifolds in integrable systems~\cite{DNOPS-2,DNOPS-1}.
\end{itemize}
See further applications in~\cite{Eynard-Lecture,DM-lectures,Borot-Lecture}.

There is a variety of generalizations of this setup, for instance, the spectral curve does not have to be compact. One can consider the so-called local topological recursion~\cite{DOSS} or blobbed topological recursion~\cite{BS-blobbed}, and all these possibilities extend the range of applications and interconnections of topological recursion.

A huge number of interesting applications demands that $x$ and $y$ while still defined on a compact spectral curves must be allowed to have logarithmic singularities. By this we mean that $dx$ and $dy$ are still globally defined meromorphic forms, but they might have non-vanishing residues at some points. Examples here include colored HOMFLY-PT polynomials of torus knots~\cite{BEM-HOMFLY,HOMFLY2,HOMFLY1}, various kinds of Hurwitz numbers~\cite{BM,Orbifold,SSZ,ACEH,bychkov2021topological}, topological string amplitudes based on topological vertex~\cite{BM,BKMP,ZhouTV,EO-3folds,Chen}, and the GW theory on ${\mathbb P}^1$~\cite{NSP1,DOSS}, just to name a few.

\subsection{\texorpdfstring{$x-y$}{x-y} swap relation}

There is a universal differential-algebraic relation that captures the exchange of $x$ and $y$ in the topological recursion procedure, and in fact it can itself replace the topological recursion procedure. It was conjectured in~\cite{borot2023functional}, proved in genus zero, under an extra assumption, in~\cite{hock2022xy}, reformulated in a much nicer way in~\cite{hock2022simple}, and finally proved in full generality in~\cite{alexandrov2022universal}.  From the point of view of $x-y$ relation, the natural meromorphic global objects are the differentials $dx$ and $dy$ rather than the functions $x$ and $y$. However, the statement that the $x-y$ relation holds~\cite[Theorem 1.14]{alexandrov2022universal} explicitly demands existence of meromorphic $x$ and $y$, thus all residues of the differentials $dx$ and $dy$ must be equal to zero.

In recent papers~\cite{hock2023laplace, hock2023xy}, Hock experimented with this formula in the presence of  logarithmic singularities of $x$ and $y$ (which then implies that $dx$ and $dy$ are allowed to have non-zero residues at the points of logarithmic singularities). In a number of special situations he conjectured necessary corrections that allow to still use the $x-y$ swap relation, and, in particular, allow to produce explicit closed formulas in the presence of logarithmic singularities.

The goal of this paper is to fully revisit the approach of Hock and put in a proper general framework his enlightening computations and conjectures. To this end we define the so-called \emph{logarithmic topological recursion} (LogTR) that associates to an input that consists of two functions $x$ and $y$ defined on a compact Riemann surface $\Sigma$, with at most logarithmic singularities (and otherwise meromorphic), a system
\begin{align}
\{{}^{\log{}}\omega^{(g)}_n\}_{\substack{g\geq 0, n\geq 1\\ 2g-2+n>0}}	
\end{align}
 of symmetric meromorphic $n$-differentials ${}^{\log{}}\omega^{(g)}_n$ on $\Sigma^n$, $n\geq 1$. Let
 \begin{align}
 	\{{}^{\log{}}\omega^{\vee,(g)}_n\}_{\substack{g\geq 0, n\geq 1\\ 2g-2+n>0}}	
 \end{align} be a system of differential associated to the same input with $x$ and $y$ interchanged.
We prove that systems of differentials $\{{}^{\log{}}\omega^{(g)}_n\}$ and $\{{}^{\log{}}\omega^{\vee,(g)}_n\}$ are still connected by the universal $x-y$ swap relation.

The present paper also prompts a comment on a recent paper by Bouchard--Kramer--Weller \cite{BKW23}, which is postponed to a subsequent publication~\cite{ABDKS5}.

Regarding the examples we have referred to at the end of Section~\ref{sec:toprecintro}, for the respective spectral curves the result of applying LogTR coincides with the result of the original TR, but the $x-y$ dual differentials are produced by the LogTR on the $x-y$ dual curve, not by the original TR (that is, on the $x-y$ dual spectral curve the result of applying LogTR differs from the original TR). In Section 3 we discuss in detail many cases where our new definition of LogTR is useful (via the $x-y$ relation).

\subsection{Blobbed topological recursion}

Blobbed topological recursion~\cite{BS-blobbed} extends the notion of usual topological recursion as a general solution of the so-called abstract loop equations~\cite{BEO-abstract-loop}. The logarithmic topological recursion is a very special instance of blobbed topological recursion with explicitly described extra singularities allowed at the points of logarithmic singularities of $y$.

There is a general explicit formula that connects blobbed topological recursion to the standard topological recursion~\cite[Equation~2.7]{BS-blobbed}. This general explicit formula can be specialized to our case and it gives an explicit expression $\{{}^{\log{}}\omega^{(g)}_n\} \longleftrightarrow \{\omega^{(g)}_n\}$ and the same on the dual side. Thus we can assemble the full system of universal relations  studied in this paper into the following diagram:
\begin{align}
	\{\omega^{(g)}_n\} \longleftrightarrow \{{}^{\log{}}\omega^{(g)}_n\}  \longleftrightarrow \{{}^{\log{}}\omega^{\vee,(g)}_n\} \longleftrightarrow \{\omega^{\vee,(g)}_n\},
\end{align}
where the first and the last arrows are re-computations with blobs adapted to our case and the middle arrow is the universal $x-y$ swap relation.

\subsection{Organization of the paper} In Section~\ref{sec:LogTR} we introduce the logarithmic topological recursion and state its compatibility with the $x-y$ swap relation. These are respectively the main definition and the main theorem of the paper.  As in the case of the original topological recursion, the $x-y$ swap relation is an especially efficient tool when on the dual side we have a trivial logarithmic topological recursion, and we survey possible applications of this technique in Section~\ref{sec:trivial-dual} (in particular, in this section we prove a refinement of the recent conjecture of Hock). We also mention independently arising cases when LogTR holds and both sides of the $x-y$ duality are nontrivial (in particular, the new Family~III of generalized double Hurwitz numbers). In Section~\ref{sec:ordinary-vs-log} we discuss the recursive formulas for the correlation differentials of logarithmic topological recursion and present a recomputation procedure that connects logarithmic topological recursion and original topological recursion (the latter one is a specialization of the formulas known for more general blobbed topological recursion). Finally, in Section~\ref{sec:proofs} we give a proof of the main theorem.

\subsection{Notation} Throughout the text we use the following notation:
\begin{itemize}
	\item $\set{n}$ denotes $\{1,\dots,n\}$.
	\item $z_I$ denotes $\{z_i\}_{i\in I}$ for $I\subseteq \set{n}$.
	\item $[u^d]$ denotes the operator that extracts the corresponding coefficient from the whole expression to the right of it, that is, $[u^d]\sum_{i=-\infty}^\infty a_iu^i \coloneqq a_d$.
	\item
	$\restr{u}{v}$ denotes the operator of substitution (or restriction), that is, $\restr{u}{v} f(u) \coloneqq f(v)$.
	\item $\cS(u)$ denotes $u^{-1}(e^{u/2} - e^{-u/2})$.
\end{itemize}

\subsection{Acknowledgments}
We are grateful to Alexander Hock and the anonymous referees for useful remarks.

 A.A. was supported by the Institute for Basic Science (IBS-R003-D1). B.B. was supported by the ISF Grant 876/20. P.D.-B. and M.K. were supported by the International Laboratory of Cluster Geometry NRU HSE, RF Government grant, ag. № 075-15-2021-608 dated 08.06.2021. S.S. was supported by the Netherlands Organization for Scientific Research. B.B., P.D.-B., M.K., and S.S. are grateful for the hospitality to IBS CGP, Pohang, where part of this research was carried out.

\section{Main definition and main theorem} \label{sec:LogTR}

\subsection{Original topological recursion}


We follow the exposition in~\cite[Section 5]{alexandrov2022universal}. Let $x$ and $y$ be meromorphic functions on $\Sigma$, and $\{\omega^{(g)}_{m}\}$, $g\geq 0$, $m\geq 1$, $2g-2+m>0$ be a given system of symmetric meromorphic differentials. We also define $\omega^{(0)}_1 \coloneqq -ydx$ and $\omega^{(0)}_2\coloneqq B$, where $B$ is the so-called Bergman kernel, that is, the unique meromorphic bi-differential with a double pole along the diagonal with bi-residue $1$ and no further singularities, whose $\mathfrak{A}$-periods are all equal to $0$.

Assume that all zeros of $dx$ are simple and $y$ is regular at the zero locus of $dx$ and the zero loci of $dx$ and $dy$ are disjoint. Denote the zeros of $dx$ by $p_1,\dots,p_N\in S$ and let $\sigma_i$ be the deck transformations of $x$ near $p_i$.

\begin{definition} We say that the system of symmetric meromorphic differentials $\{\omega^{g}_m\}$, $g\geq 0$, $m\geq 1$ satisfies
\begin{itemize}
		\item the \emph{linear loop equations}, if for any $g,m\geq 0$, $i=1,\dots,N$
		\begin{equation} \label{eq:LLE-original}
			\omega^{(g)}_{m+1}(z_{\llbracket m \rrbracket},z) + \omega^{(g)}_{m+1}(z_{\llbracket m \rrbracket},\sigma_i(z))
		\end{equation}
		is holomorphic at $z\to p_i$ and has at least a simple zero in $z$ at $z=p_i$;
		\item the \emph{quadratic loop equations}, if for any $g,m\geq 0$, $i=1,\dots,N$, the quadratic differential in $z$
		\begin{equation} \label{eq:QLE-original}
			\omega^{(g-1)}_{m+2}(z_{\llbracket m \rrbracket},z,\sigma_i(z)) + \sum_{\substack{g_1+g_2=g\\ I_1\sqcup I_2 = \llbracket m \rrbracket
			}} \omega^{(g_1)}_{|I_1|+1} (z_{I_1},z)\omega^{(g_2)}_{|I_2|+1} (z_{I_2},\sigma_i(z))
		\end{equation}
		is holomorphic at $z\to p_i$ and has at least a double zero at $z=p_i$.
\end{itemize}
\end{definition}

\begin{definition} We say that $\{\omega^{(g)}_{m}\}$ satisfies the \emph{blobbed topological recursion}, if it satisfies the linear and quadratic loop equations for each $g,m\geq 0$.
\end{definition}

Under the assumption that all zeros $\{p_1,\dots,p_N\}$ of $dx$ are simple, the statement that $\{\omega^{(g)}_{m}\}$ satisfies the blobbed topological recursion implies that the principal parts near the zero locus of $dx$ of each $\omega^{(g)}_{m}$, $2g-2+m>0$, in each variable are uniquely determined by $\omega^{(g')}_{m'}$ with $2g'-2+m' < 2g-2+m$.

\begin{definition} We say that $\{\omega^{(g)}_{m}\}$ satisfy the original \emph{topological recursion} (TR) on the spectral curve $(\Sigma,x,y)$, if in addition to the blobbed topological recursion they satisfy the so-called \emph{projection property}: for any $g\geq 0$, $m\geq 1$, $2g-2+m>0$
	\begin{equation}
		\omega^{(g)}_{m}(z_{\set{m}}) = \sum_{i_1,\dots,i_m=1}^N \Bigg(\prod_{j=1}^m \res_{z_j'\to p_{i_j}} \int^{z_j'}_{p_{i_j}} B(\cdot,z_j)\Bigg) \omega^{(g)}_{m}(z'_{\set{m}}).
	\end{equation}
\end{definition}

In other words, we say that in the case of original topological recursion the meromorphic differentials $\omega^{(g)}_{m}$, $2g-2+n>0$, have no other poles than at $\{p_1,\dots,p_N\}$, and all their $\mathfrak{A}$-periods are equal to zero.

\subsection{Logarithmic topological recursion} Topological recursion relations depend on the differentials $dx$, $dy$ only, but not the functions $x$ and $y$ themselves, and they are well defined if $dx$ and $dy$ are global meromorphic differentials with possibly nonzero residues. In the latter case the functions $x$ and $y$ are not univalued and posses logarithmic singularities. In the presence of this sort of singularities it is useful to consider, along with the original topological recursion, its variation that we call logarithmic topological recursion.

Assume $dx$ and $dy$ are globally defined meromorphic $1$-forms on a compact Riemann surface $\Sigma$ with possibly non-vanishing residues at some points. We still assume that all zeros $\{p_1,\dots,p_N\}$ of $dx$ are simple, $dy$ is regular at the zero locus of $dx$ and the zero loci of $dx$ and $dy$ are disjoint.

\begin{definition}
We say that the primitive $y$ of the differential $dy$ posses \emph{logarithmic singularity} at some point~$a$ on~$\Sigma$ if $dy$ has a pole at~$a$ with nonzero residue. A logarithmic singularity of~$y$ is called \emph{LogTR-vital} if this pole of $dy$ is simple and $dx$ has no pole at this point.
\end{definition}

Let $a_1,\dots,a_M$ be the LogTR-vital singular points of $y$. We denote the residues of~$dy$ at these points by $\alpha_1^{-1},\dots,\alpha_M^{-1}$, respectively. That is, the principal part of $dy$ near $a_i$ is given by $\alpha_i^{-1} dz / (z-a_i)$ in any local coordinate~$z$.

Consider the germs of the principal parts of meromorphic $1$-forms in the neighborhoods of $a_i$, $i=1,\dots,M$ defined as the principal parts of coefficients of positive powers of $\hbar$ in the following expressions:
\begin{align}\label{eq:PrincipalPartsLog}
\left( \frac{1}{\alpha_i\cS(\alpha_i\hbar \partial_{x})}\log(z-a_i)\right)dx.
\end{align}

\begin{definition}\label{def:log-proj} We say that $\{{}^{\log{}}\omega^{(g)}_{n}\}$ satisfies the \emph{logarithmic topological recursion} (LogTR) on the spectral curve $(\Sigma,x,y)$, if ${}^{\log{}}\omega^{(0)}_1 = -ydx$, ${}^{\log{}}\omega^{(0)}_2 = B$, $\{{}^{\log{}}\omega^{(g)}_{n}\}$ satisfy the blobbed topological recursion and additionally they satisfy the so-called \emph{logarithmic projection property}: for any $g\geq 0$, $m\geq 1$, $2g-2+m>0$
	\begin{align} \label{eq:LogProjection}
		{}^{\log{}}\omega^{(g)}_{m}(z_{\set{m}}) & = \sum_{i_1,\dots,i_m=1}^N \Bigg(\prod_{j=1}^m \res_{z_j'\to p_{i_j}} \int^{z_j'}_{p_{i_j}} B(\cdot,z_j)\Bigg) {}^{\log{}}\omega^{(g)}_{m}(z'_{\set{m}})
		\\ \notag & \quad
		- \delta_{m,1}[\hbar^{2g}] \sum_{i=1}^M  \res_{z'\to a_i} \Bigg(\int^{z'}_{a_i} B(\cdot,z_1)\Bigg)
\left(\frac{1}{\alpha_i\cS(\alpha_i\hbar \partial_{x'})}\log(z'-a_i)\right)dx'.
	\end{align}
\end{definition}

In other words, we say that in the case of logarithmic topological recursion the meromorphic differentials ${}^{\log{}}\omega^{(g)}_{m}$, $2g-2+m>0$, $m\geq 2$, have no other poles than at $\{p_1,\dots,p_N\}$  and all their $\mathfrak{A}$-periods are equal to zero.
For $g\geq 1$ and $m=1$ we have a bit different setup: they
have poles at $\{p_1,\dots,p_N\}$, and also at the points $\{a_1,\dots,a_M\}$. The principal parts at the latter points are equal to the principal parts of expressions given by~\eqref{eq:PrincipalPartsLog}, and all their $\mathfrak{A}$-periods are equal to zero.

\begin{remark} Note that if $a$ is a simple pole of~$dy$ and $dx$ has also a pole at this point, that is, singularity of $y$ is not LogTR-vital, then 
the principal part of~$[\hbar^{2g}]\bigl( \frac{1}{\alpha\cS(\alpha\hbar \partial_{x})}\log(z-a)\bigr)dx$ with $g>0$ is equal to zero,
so we could formally include these points to the list of LogTR-vital singularities of~$y$ without changing relation of recursion. But it is important to underline that such singular points do not contribute to the recursion and do not affect Equation~\eqref{eq:LogProjection}.
\end{remark}

\begin{remark}
The poles of order greater than $1$ are excluded from the logarithmic projection property by similar local considerations. Indeed, assume that the meromorphic form~$dy$ varies in a family such that two simple poles collapse together producing a pole of order~$2$. Then, it is easy to see by local computations that the residues at these two poles tend to infinity when the points approach one another. 

The model situation for the local computation is as follows. Note that 
\begin{align}
\frac{(\alpha^{-1}(\epsilon)+c)\,dz}{z-\epsilon}-\frac{\alpha^{-1}(\epsilon)\,dz}{z+\epsilon} \to \frac{dz}{z^2} +\frac{c\, dz}{z}\qquad \text{for} \qquad \epsilon\to 0
\end{align}
means that $2\alpha^{-1}(\epsilon)\epsilon\to 1$ for $\epsilon\to 0$. Hence $\alpha(\epsilon)\to 0$ for $\epsilon\to 0$, and, furthermore,
\begin{align}
	& d\,\lim_{\epsilon\to 0} \bigg(\frac{1}{\frac{\alpha(\epsilon)}{1+c\,\alpha(\epsilon)}\cS(\frac{\alpha(\epsilon)}{1+c\,\alpha(\epsilon)}\hbar \partial_{z})}\log(z-\epsilon) + \frac{1}{-\alpha(\epsilon)\cS(-\alpha(\epsilon)\hbar \partial_{z})}\log(z+\epsilon) \bigg) 
	\\ \notag & = \lim_{\epsilon\to 0} \bigg( \frac{(\alpha^{-1}(\epsilon)+c)\,dz}{z-\epsilon}-\frac{\alpha^{-1}(\epsilon)\,dz}{z+\epsilon} \bigg) = \frac{dz}{z^2} +\frac{c\, dz}{z},
\end{align}
so the pole of order $2$ does not contribute. 

It follows that the limiting contribution of these two points to~\eqref{eq:PrincipalPartsLog} is trivial, independently of whether the limiting pole of order two has a residue or not. The same argument also works for the higher order poles.
\end{remark}

\subsection{\texorpdfstring{$x-y$}{x-y} swap relation for LogTR} One of the ways to motivate the definition of logarithmic topological recursion is to see that this extension survives the general $x-y$ swap relation that was established in the absence of logarithmic singularities.

Assume we have completely symmetric assumptions for $x$ and $y$: both have meromorphic differentials, the zeros of $dx$ and $dy$ are simple and disjoint, $dy$ is regular at the zero locus of $dx$ and vice versa.

Then we have two instances of logarithmic topological recursion on $\Sigma$: for the input  given by the pair $(x,y)$ and for the input given by $(x^\vee, y^\vee)\coloneqq (y,x)$. Denote the meromorphic differentials produced by the latter input by $\{{}^{\log{}}\omega^{\vee,(g)}_n\}$.

In the following theorem we use graphs with multiedges (we follow~\cite{alexandrov2022universal}). A graph with multiedges is given by its set of vertices $V$, its set of multiedges $E$, the set of red flags $R$, and the set of blue flags $L$, equipped with maps $r\colon R\to V$, $l\colon L\to E$, and an isomorphism $\iota \colon L\to R$. 
The index of a vertex $v\in V$ is $|r^{-1}(v)|$, the index of a multiedge $e\in E$ is $|e|:=|l^{-1}(e)|$. We call elements of $l^{-1}(e)$ the \emph{legs} of a multiedge $e$, and we say that a particular leg $h$ of $e$ is attached to a vertex $v$ if $v=r(\iota(h))$. 

\begin{theorem} \label{thm:LogXYSwap} We have:
		\begin{align} \label{eq:MainFormulaSimple}
		& \frac{{}^{\log{}}\omega_{n}^{\vee,(g)} (z_{\llbracket n\rrbracket})}{\prod_{i=1}^n dx_i^\vee} 
		=
		(-1)^n
		\coeff \hbar {2g} \sum_{\Gamma} \frac{\hbar^{2g(\Gamma)}}{|\mathrm{Aut}(\Gamma)|} \prod_{i=1}^n
		\sum_{k_i=0}^\infty  \partial_{y_i}^{k_i} [u_i^{k_i}] \frac{dx_i}{dy_i}
		\\ \notag &
		\frac{1}{u_i} e^{ u_i \cS(\hbar u_i \partial_{x_i}) \sum_{\tilde g=0}^\infty \hbar^{2\tilde g} \frac{{}^{\log{}}\omega^{(\tilde g)}_{1} (z_i)}{dx_i}-u_i \frac{{}^{\log{}}\omega^{(0)}_{1} (z_i)}{dx_i}}
		\\ \notag &
		\prod_{e\in E(\Gamma)} \prod_{j=1}^{|e|\geq 2}\restr{(\tilde u_j, \tilde x_j)}{ (u_{e(j)},x_{e(j)})} \tilde u_j \cS(\hbar \tilde u_j  \partial_{\tilde x_j}) \sum_{\tilde g=0}^\infty \hbar^{2\tilde g}\, \frac{{}^{\log{}}\tilde\omega^{(\tilde g)}_{|e|}(\tilde z_{\llbracket |e|\rrbracket})}{\prod_{j=1}^{|e|} d\tilde x_j}
		\\ \notag &
		+\delta_{(g,n),(0,1)} (-x_1).
	\end{align}
Here
\begin{itemize}
	
	\item The sum is taken over all connected graphs $\Gamma$ with $n$ labeled vertices (with labels from $1$ to $n$) and multiedges of index $\geq 2$. 
	\item For convenience, for a given such graph, we also label all legs of every given multiedge $e$ from $1$ to $|e|$ in an arbitrary way.
	\item For a multiedge $e$ with index $|e|$  we control its attachment to the vertices by the associated map $e\colon \llbracket |e| \rrbracket \to 
	\llbracket n \rrbracket
	$ that we denote also by $e$, abusing notation (so $e(j)$ is the label of the vertex to which the $j$-th leg of the multiedge $e$ is attached). Do note that this map can be an arbitrary map from $\llbracket |e| \rrbracket$ to $\llbracket n \rrbracket$; in particular, it might not be injective, i.e. we allow a given multiedge to connect to a given vertex with several of its legs. 
	\item For a given multiedge $e$ with $|e|=2$ we define ${}^{\log{}}\tilde \omega^{(0)}_{2}(\tilde x_1,\tilde x_2) \coloneqq  {}^{\log{}}\omega^{(0)}_{2}(\tilde x_1,\tilde x_2) - \frac{d\tilde x_1d\tilde x_2}{(\tilde x_1-\tilde x_2)^2}$ if $e(1)=e(2)$, and ${}^{\log{}}\tilde \omega^{(0)}_{2}(\tilde x_1,\tilde x_2) =  {}^{\log{}}\omega^{(0)}_{2}(\tilde x_1,\tilde x_2)$ otherwise. For all $(g,n)\not=(0,2)$ we simply have $\tilde \omega^{(g)}_{n,0} :=  \omega^{(g)}_{n,0}$.
	\item By $g(\Gamma)$ we denote the first Betti number of $\Gamma$.
	\item $|\mathrm{Aut}(\Gamma)|$ stands for the number of automorphisms of $\Gamma$.
\end{itemize}
\end{theorem}

\begin{remark} The statement of this theorem repeats literally the statement of~\cite[Theorem 1.14]{alexandrov2022universal}, which was done for the original topological recursion (that is, under an extra assumption that $x$ and $y$ are meromorphic rather than $dx$ and $dy$).
\end{remark}

\begin{remark} \label{rem:analyticbeh} The analytic behavior of the stable differentials participating in Theorem~\ref{thm:LogXYSwap} is as follows. We have an original system of differential $\{{}^{\log{}}\omega^{(g)}_n\}$ satisfying LogTR. It means that these differentials have poles at zeros of $dx$, and, if LogTR-vital singularities of~$y$ are present, at these points as well (for $n=1$ only). The statement of Theorem~\ref{thm:LogXYSwap} implies that all these poles cancel on the right hand side of~\eqref{eq:MainFormulaSimple}. Namely, the cancellation of these poles at zeros of~$dx$ is implied by the ordinary $x-y$ swap statement. Theorem claims that the poles at the poles of~$dy$ also cancel if the differentials satisfy the relation of topological recursion with additional account of contribution of logarithmic singularities as in~\eqref{eq:PrincipalPartsLog} and~\eqref{eq:LogProjection}. Instead, the differentials on the left hand side attain singularities at the zeros of~$dy$ due to presence of operators $\partial_{y_i}$, and also in certain cases at the poles of $dy^\vee=dx$, namely, only if~$n=1$ and if these poles of~$dy^\vee$ are LogTR-vital.
\end{remark}
	
\begin{remark}
This theorem by itself justifies the definition of logarithmic topological recursion. Note that in the ordinary case one re-defines~\cite{alexandrov2022universal} original topological recursion as a procedure that gives a system of symmetric meromorphic differentials with prescribed type of singularities depending on $dx$ that under the transformation~\eqref{eq:MainFormulaSimple} turns into the system of differentials with prescribed type of singularities depending on $dx^\vee\coloneqq dy$.

In the logarithmic case one can re-define logarithmic topological recursion  as a procedure that gives a system of symmetric meromorphic differentials with prescribed type of singularities depending on $dx$ and $dy$ that under the transformation~\eqref{eq:MainFormulaSimple} turns into the system of differentials with prescribed type of singularities depending on $dx^\vee\coloneqq dy$ and $dy^\vee\coloneqq dx$.
\end{remark}

\begin{remark}
If either $x$ or $y$ is meromorphic, then either $\{{}^{\log{}}\omega^{\vee,(g)}_n\}$ or, respectively, $\{{}^{\log{}}\omega^{(g)}_m\}$ satisfy the original TR that might be considered as a special case of LogTR.
\end{remark}

\begin{remark}\label{rmk213}
If the sets of logarithmic singularities of $x$ and $y$ are the subsets of singularities (not necessarily logarithmic) of $y$ and $x$, respectively, then all singularities are not LogTR-vital and we have on both sides of the $x-y$ swap relation~\eqref{eq:MainFormulaSimple} the differentials produced by the original topological recursion, that is $\{{}^{\log{}}\omega^{\vee,(g)}_n\}=\{\omega^{\vee,(g)}_n\}$ and $\{{}^{\log{}}\omega^{(g)}_m\}=\{\omega^{(g)}_m\}$,  so in this special case Theorem~\ref{thm:LogXYSwap} turns into~\cite[Theorem 1.14]{alexandrov2022universal} with relaxed assumptions on $x$ and $y$.
\end{remark}

\begin{remark} An alternative way to formulate Theorem~\ref{thm:LogXYSwap} is to give a relation expressing $\{{}^{\log{}}\omega_{n}^{(g)}\}$ in terms of  $\{{}^{\log{}}\omega_{n}^{\vee,(g)}\}$. The expression is then exactly the same with $x$ and $y$ replaced by $x^\vee$ and $y^\vee$, respectively:
\begin{align} \label{eq:MainFormulaInverse}
	& \frac{{}^{\log{}}\omega_{n}^{(g)} (z_{\llbracket n\rrbracket})}{\prod_{i=1}^n dx_i} 	=
	(-1)^n
	\coeff \hbar {2g} \sum_{\Gamma} \frac{\hbar^{2g(\Gamma)}}{|\mathrm{Aut}(\Gamma)|} \prod_{i=1}^n
	\sum_{k_i=0}^\infty  \partial_{y^\vee_i}^{k_i} [u_i^{k_i}] \frac{dx^\vee_i}{dy^\vee_i}
	\\ \notag &
	\frac{1}{u_i} e^{ u_i \cS(\hbar u_i \partial_{x^\vee_i}) \sum_{\tilde g=0}^\infty \hbar^{2\tilde g} \frac{{}^{\log{}}\omega^{\vee,(\tilde g)}_{1} (z_i)}{dx^\vee_i}-u_i \frac{{}^{\log{}}\omega^{\vee,(0)}_{1} (z_i)}{dx^\vee_i}}
	\\ \notag &
	\prod_{e\in E(\Gamma)} \prod_{j=1}^{|e|\geq 2}\restr{(\tilde u_j, \tilde x^\vee_j)}{ (u_{e(j)},x^\vee_{e(j)})} \tilde u_j \cS(\hbar \tilde u_j  \partial_{\tilde x^\vee_j}) \sum_{\tilde g=0}^\infty \hbar^{2\tilde g}\, \frac{{}^{\log{}}\tilde\omega^{\vee,(\tilde g)}_{|e|}(\tilde z_{\llbracket |e|\rrbracket})}{\prod_{j=1}^{|e|} d\tilde x^\vee_j}
	\\ \notag &
	+\delta_{(g,n),(0,1)} (-x^\vee_1).
\end{align}
\end{remark}

\section{The case of trivial dual LogTR} \label{sec:trivial-dual}

\subsection{Trivial dual logarithmic topological recursion}
In this section we provide a number of applications of Theorem~\ref{thm:LogXYSwap} (to be precise, we consider its dual formulation \eqref{eq:MainFormulaInverse}) to the cases when $y$ is unramified. A number of such applications in the case when $x$ is rational are listed in~\cite{alexandrov2022universal}. Here we concentrate on those cases when $x$ possess logarithmic singularities and the ordinary $x-y$ swap formula for the original topological recursion may not work. This was studied by Hock~\cite{hock2023laplace,hock2023xy}, and he conjectured in a number of cases the natural corrections (which reproduced some known cases studied earlier in~\cite{bychkov2021topological}), as well as a general principle of how it should work.

The duality formula~\eqref{eq:MainFormulaInverse} simplifies dramatically when the dual topological recursion is trivial, that is, when $dy$ has no zeros and $\omega^{\vee,(g)}_n=0$ for $2g-2+n>0$.
In principle, the formulas use not $y$, but $dy$, so we get extra possible cases of $dy$ being a holomorphic differential on an elliptic curve. These cases will be discussed elsewhere in subsequent papers. In the present paper we always assume that $\Sigma=\mathbb{C}P^1$. By Riemann--Hurwitz formula, $dy$ has two poles, and there are essentially two distinct cases:
\begin{itemize}
\item I: the two poles of $dy$ are coinciding. Then $y$ itself is rational and has degree one, that is, $y=z$ can be taken as an affine coordinate on~$\Sigma$
\item II: the two poles of $dy$ are distinct.  Then it is given by $dy=\frac{dz}{z}$ for a suitable affine coordinate~$z$, up to a constant factor. This constant factor can be absorbed by the rescaling of the correlation differentials or $x$, so we can put $y=\log(z)$ in this case.
\end{itemize}

In both cases, let us choose an affine coordinate~$z$ such that $z=\infty$ is a pole of $dy$ (a choice of $z$ such that $y=z$ or $y=\log z$ fits the requirement). Let us represent the function~$x$ in the form
\begin{equation}\label{eq:x_with_log}
x(z)=\sum_{i=1}^M\frac{\log(z-a_i)}{\alpha_i}+R(z),
\end{equation}
where $a_1,\dots,a_M$ are the LogTR-vital singularities of~$x$ and $1/\alpha_i$ are the corresponding residues of~$dx$. The term $R(z)$ may contain contributions of other non-LogTR-vital logarithmic singularities, so we do not assume that $R$ is rational. Set
\begin{align} \label{eq:hat-x}
\hat x(z,\hbar)&=\sum_{i=1}^M\frac1{\cS(\alpha_i\hbar\partial_{y_i})}\frac{\log(z-a_i)}{\alpha_i}+R(z),\\
\label{eq:wigeneral}
w_i&=e^{\frac{u_i\hbar}{2}\partial_{y_i}}z_i,\quad
\bar w_i=e^{-\frac{u_i\hbar}{2}\partial_{y_i}}z_i.
\end{align}

As a special case of Theorem~\ref{thm:LogXYSwap}, or, more precisely, Equation~\eqref{eq:MainFormulaInverse}, we obtain

\begin{proposition}\label{prop:DualTRTrivial}
If either $y$ or $e^y$ is an affine coordinate on $\Sigma=\C P^1$, then the differentials of LogTR with the initial differential ${}^{\log{}}\omega^{(0)}_1=-y\,dx$ are given by the following explicit formula
\begin{multline}\label{eq:main}
{}^{\log{}}\omega_n^{(g)}(z_{\set{n}})=
[\hbar^{2g-2+n}](-1)^n\prod_{i=1}^n\left(\sum_{r=0}^\infty \bigl(d_i\tfrac1{dx_i}\bigr)^r \;
[u_i^r]\,e^{-u_i(\cS(u_i\hbar\partial_{y_i})\hat x_i-x_i)}\right)
\\(-1)^{n-1}\sum_{\sigma\in\{n\text{-cycles}\}}\prod_{i=1}^n\frac{\sqrt{dw_id\bar w_{\sigma(i)}}}{w_i-\bar w_{\sigma(i)}}.
\end{multline}
Here summation runs over the set of $n$-cycles in the permutation group $S_n$, by $\sqrt{dw_id\bar w_i}$ we mean $\sqrt{\frac{dw_i}{dz_i}\frac{d\bar w_i}{dz_i}}dz_i$,  and $d\frac1{dx}$ is the operator acting in the space of meromorphic differentials and taking a form $\eta$ to $d\frac{\eta}{dx}$.
\end{proposition}

\begin{remark}
The functions $w_i$ and $\bar w_i$ entering the formula can be given more explicitly in the two cases of a choice~for $y$ as follows:
\begin{align}\label{eq:wiexplicit}
y&=z:&e^{\pm\frac{\hbar u}{2}\partial_y}z&=z\pm\frac{u\hbar}{2},
\\y&=\log z:&e^{\pm\frac{\hbar u}{2}\partial_y}z&=e^{\pm\frac{u\hbar}{2}}z.
\end{align}
In both cases we have $\sqrt{dw_id\bar w_i}=dz_i$ and
\begin{equation}\label{eq:cdetz}
(-1)^{n-1}\sum_{\sigma\in\{n\text{-cycles}\}}\prod_{i=1}^n\frac{\sqrt{dw_id\bar w_{\sigma(i)}}}{w_i-\bar w_{\sigma(i)}}=
(-1)^{n-1}\sum_{\sigma\in\{n\text{-cycles}\}}\prod_{i=1}^n\frac{1}{w_i-\bar w_{\sigma(i)}}\;
\prod_{i=1}^n dz_i.
\end{equation}
However, we keep an expression for the corresponding connected determinant in Eq.~\eqref{eq:main} in the form of the left hand side of~\eqref{eq:cdetz} since this form is invariant with respect to a choice of an affine coordinate on~$\Sigma=\C P^1$, see Remark~\ref{rem:choice_for_z} below.
\end{remark}

\begin{remark}\label{rem:choice_for_z}
It is also possible to use an arbitrary affine coordinate~$z$. In this case, one should replace $\log(z-a_i)$ by $\log\bigl(\tfrac{z-a_i}{z-p_0}\bigr)$ everywhere in~\eqref{eq:x_with_log} and~\eqref{eq:hat-x}, where $p_0$ is a chosen pole of $dy$, one and the same for all logarithmic summands.
\end{remark}

\begin{remark} \label{rem:Appl-dual-Log-Ord}
In many applications the LogTR for the differentials of Proposition coincides with the original one, and $\{{}^{\log{}}\omega^{(g)}_n\}=\{\omega^{(g)}_n\}$. In particular, the two recursions coincide in the case when $y=z$ is an affine coordinate, since $dy$ has no simple poles in this case. In the case $y=\log z$ the two recursions also coincide if $dx$ has poles both at $z=0$ and $z=\infty$, see Remark~\ref{rmk213}. The only case when we should make a distinction between logarithmic and original topological recursions in the setting of Proposition is when $y=\log z$ and $dx$ is regular either at $z=0$ or at $z=\infty$. 
\end{remark}

\begin{remark} 
	Formula~\eqref{eq:main} is equivalent to the determinantal formula discussed in~\cite{ABDKS3} for this case. To see it, one needs to take the determinantal formula for $\omega_n$ from~\cite[Section~3.2]{ABDKS3}, substitute there the formula from~\cite[Corollary~3.14]{ABDKS3}, and take into account~\cite[Lemma~6.3]{ABDKS3}. Due to the appearance of logarithmic singularities, one needs to replace $x$ with $\hat x$ in the integrals in the exponents in the second line of the formula of~\cite[Corollary~3.14]{ABDKS3}. 
\end{remark}

\begin{proof}[Proof of Proposition~\ref{prop:DualTRTrivial}] This proposition is a direct corollary of Equation~\eqref{eq:MainFormulaInverse}. It is more convenient to compare expressions for the corresponding \emph{disconnected} $n$-point differentials. The latter ones can be defined formally, we just assemble ${}^{\log{}}\omega_n^{(g)}(z_{\set{n}})$ into the generating series ${}^{\log{}}\omega_n(z_{\set{n}})\coloneqq \sum_{g=0}^\infty \hbar^{2g-2+n} {}^{\log{}}\omega_n^{(g)}(z_{\set{n}})$ and then consider the sums 
\begin{align}
	{}^{\log{}}\omega_n^\bullet (z_{\set{n}})\coloneqq \sum_{k=1}^n \sum_{I_1\sqcup \cdots \sqcup I_k = \set{n}} \prod_{j=1}^k {}^{\log{}}\omega_{|I_j|}(z_{I_j}). 	
\end{align}	
Applying this to~\eqref{eq:MainFormulaInverse} and~\eqref{eq:main}, we see that the sum on the right hand side of~\eqref{eq:MainFormulaInverse} is replaced by the corresponding sum over the set of all graphs on $n$ marked vertices, not necessarily connected, and the sum on the right hand side of~\eqref{eq:main} is replaced by the corresponding determinant.

Consider the dual LogTR for $x^\vee = y$ and $y^\vee = x$. Since $dx^\vee(z)$ does not have any zeros, the stable correlation differentials of the original (as opposed to logarithmic) TR on the curve $(\Sigma, x^{\vee},y^{\vee})$ are all equal to zero.

In the case of LogTR, the principal parts at points of the logarithmic singularities of $y^\vee$ determine all $\omega^{\vee,(g)}_1$, $g\geq 1$, and comparing Equations~\eqref{eq:LogProjection} and~\eqref{eq:hat-x} we see that
\begin{gather} \label{eq:x-hat-omega-1}
\frac{{}^{\log{}}\omega^{\vee, (g)}_1(z)}{dx^\vee(z)}
=[\hbar^{2g}](-\hat x),\\
\label{eq:VertexContributions}
e^{ u_i \cS(\hbar u_i \partial_{x^\vee_i}) \sum_{\tilde g=0}^\infty \hbar^{2\tilde g} \frac{{}^{\log{}}\omega^{\vee,(\tilde g)}_{1} (z_i)}{dx^\vee_i}-u_i \frac{{}^{\log{}}\omega^{\vee,(0)}_{1} (z_i)}{dx^\vee_i}}
=e^{-u_i(\cS(u_i\hbar\partial_{y_i})\hat x_i-x_i)}.
\end{gather}

For $n>1$, the stable LogTR differentials ${}^{\log{}}\omega^{\vee,(g)}_n$ are still trivial and the contributions of all multiedges to~\eqref{eq:MainFormulaInverse} are equal to zero except the contributions of the edges with two legs $e(1)=i$, $e(2)=j$ for various pairs $(i,j)$ corresponding to the differentials
\begin{equation}
{}^{\log{}}\tilde\omega^{\vee,(0)}_{2}(\tilde z_1,\tilde z_2)=
\begin{cases}
\frac{d\tilde z_1d\tilde z_2}{(\tilde z_1-\tilde z_2)^2},&i\ne j,\\
\frac{d\tilde z_1d\tilde z_2}{(\tilde z_1-\tilde z_2)^2}-
\frac{d\tilde y_1d\tilde y_2}{(\tilde y_1-\tilde y_2)^2},&i=j.
\end{cases}
\end{equation}
In the case $i\ne j$, the contribution of such an edge is equal to
\begin{align}
\mathfrak{w}_{i,j} & :=\hbar u_i\cS(u_i\hbar\partial_{x_i})\hbar u_j\cS(u_j\hbar\partial_{x_j})\frac{\omega^{\vee,(0)}_{2}(z_i,z_j)}{dx_idx_j}
\\ \notag & =(e^{\frac{u_i\hbar}{2}\partial_{x_i}}-e^{-\frac{u_i\hbar}{2}\partial_{x_i}})
(e^{\frac{u_j\hbar}{2}\partial_{x_j}}-e^{-\frac{u_j\hbar}{2}\partial_{x_j}})
\log(z_i-z_j)
=\log\frac{(w_i-w_j)(\bar w_i-\bar w_j)}
{(\bar w_i-w_j)(w_j-\bar w_i)},
\end{align}
where $w_i$, $\bar w_i$ are defined by~\eqref{eq:wigeneral}. Since multiple edges are allowed, we can compute the contribution of all multiple edges connecting the two given different vertices $i$ and $j$ as
\begin{equation}\label{eq:EdgeContributions}
e^{\mathfrak{w}_{i,j}}=\frac{(w_i-w_j)(\bar w_i-\bar w_j)}
{(\bar w_i-w_j)(w_j-\bar w_i)}.
\end{equation}

In the case $i=j$, one can show that the total contribution of multiple loops (edges with equal legs $e(1)=e(2)=i$) is equal to
\begin{equation}\label{eq:LoopContributions}
e^{\frac12
	\restr{\tilde z_1}{z_i} \restr{\tilde z_2}{z_i} \hbar^2 u_i^2\cS(u_i\hbar\partial_{\tilde y_1})\cS(u_i\hbar\partial_{\tilde y_2})
	\left(\frac{\frac{d\tilde z_1}{d\tilde y_1}\frac{d\tilde z_2}{d\tilde y_2}}{(\tilde z_1-\tilde z_2)^2}-\frac{1}{(\tilde y_1-\tilde y_2)^2}\right)} = \frac{\hbar u_i\sqrt{\frac{dw_i}{dy_i}\frac{d\bar w_i}{dy_i}}}{w_i-\bar w_i}.
\end{equation}
Indeed, 
\begin{align}\label{eq:Somega02reg}
	& \frac 12 \restr{z_1}{z} \restr{z_2}{z} \hbar^2 u^2\cS(u\hbar\partial_{y_1})\cS(u\hbar\partial_{y_2})
	\left(\frac{\frac{dz_1}{dy_1}\frac{dz_2}{dy_2}}{(z_1-z_2)^2}-\frac{1}{(y_1-y_2)^2}\right)
	 \\ \notag &
	\quad =\frac 12  \int\limits_{\bar w}^{w}\int\limits_{\bar w}^{w} d_1 d_2 \log\Big(\frac{z_1-z_2}{y_1-y_2}\Big)
		=\frac 12 \log\left(\left(\frac{y(w)-y(\bar w)}{w-\bar w}\right)^2\frac{dwd\bar w}{dy(w)dy(\bar w)}\right)
		\\ \notag & \quad
		=\frac 12 \log\left(\left(\frac{\hbar u }{w-\bar w}\right)^2\frac{dwd\bar w}{dydy}\right)
		=\log\left(\frac{\hbar u \sqrt{\frac{dwd\bar w}{dydy}}}{w-\bar w}\right)
		.
\end{align}

Substituting \eqref{eq:VertexContributions},~\eqref{eq:EdgeContributions},~\eqref{eq:LoopContributions} into \eqref{eq:MainFormulaInverse}, we get
\begin{align}
	&
	(-1)^n
	\coeff \hbar {2g} \sum_{\widetilde{\Gamma}} \frac{\hbar^{2g(\widetilde{\Gamma})}}{|\mathrm{Aut}(\widetilde{\Gamma})|} \prod_{i=1}^n
	\sum_{k_i=0}^\infty  \partial_{y^\vee_i}^{k_i} [u_i^{k_i}] \frac{dx^\vee_i}{dy^\vee_i}
	\\ \notag &
	\frac{1}{u_i} e^{-u_i(\cS(u_i\hbar\partial_{y_i})\hat x_i-x_i)}
	\\ \notag &
	\prod_{e\in E(\widetilde{\Gamma})} \prod_{j=1}^{|e|\geq 2}\restr{(\tilde u_j, \tilde x^\vee_j)}{ (u_{e(j)},x^\vee_{e(j)})} \tilde u_j \cS(\hbar \tilde u_j  \partial_{\tilde x^\vee_j}) \sum_{\tilde g=0}^\infty \hbar^{2\tilde g}\, \frac{{}^{\log{}}\tilde\omega^{\vee,(\tilde g)}_{|e|}(\tilde z_{\llbracket |e|\rrbracket})}{\prod_{j=1}^{|e|} d\tilde x^\vee_j}
	\\ \notag &
	+\delta_{(g,n),(0,1)} (-x^\vee_1).
\end{align}

Comparing~\eqref{eq:MainFormulaInverse} with~\eqref{eq:main} and using~\eqref{eq:VertexContributions},~\eqref{eq:EdgeContributions},~\eqref{eq:LoopContributions}, we see that (the disconnected version of) Eq.~\eqref{eq:main} follows from the well known Cauchy determinant formula
\begin{equation}
\det\left|\frac{1}{w_i-\bar w_j}\right|_{i,j=1,\dots,n}=
\prod_{i=1}^n\frac{1}{w_i-\bar w_i}\;\prod_{1\le i<j\le n}
\frac{(w_j-w_i)(\bar w_i-\bar w_j)}
{(w_i-\bar w_j)(w_j-\bar w_i)},
\end{equation}
and it remains only to take its connected counterpart.
\end{proof}

Let us consider several examples of logarithmic topological recursion when one can apply Proposition \ref{prop:DualTRTrivial}.

\subsection{Lambert spectral curve and linear Hodge integrals}\label{SHodge}

The Lambert spectral curve data are
\begin{equation}\label{eq:Lambert-z}
x=z-\log(z),\quad y=z,
\end{equation}
where $z$ is an affine coordinate on~$\Sigma=\C P^1$.

The differential $dx=\frac{z-1}{z}dz$ has a unique zero at $z=1$. The differentials of this topological recursion are related to enumeration of Hurwitz numbers and also to linear Hodge integrals over the moduli spaces. Namely, consider the sequence of rational functions defined by
\begin{equation}\label{eq:xi}
\xi_k(z)=(-\partial_x)^{k+1}y,\quad k=0,1,2,\dots.
\end{equation}
Then the differentials of TR with the spectral curve data~\eqref{eq:Lambert-z} are given by \cite{BM,MulaseCaJ}
\begin{equation}\label{eq:linearHodge}
\omega^{(g)}_n(z_{\set{n}})=\sum_{k_1+\dots+k_n+j=3g-3+n}(-1)^j
\int\limits_{\overline{\mathcal M}_{g,n}}\lambda_j\,\psi_1^{k_1}\dots\psi_n^{k_n}\;d\xi_{k_1}(z_1)\dots d\xi_{k_n}(z_n).
\end{equation}

The point $z=0$ is a LogTR-vital logarithmic singularity of~$x$ and Proposition~\ref{prop:DualTRTrivial} says that these differentials can also be computed explicitly by~\eqref{eq:main} with
\begin{equation}
\hat x=z-\frac{1}{\cS(\hbar\partial_z)}\log(z)=\frac{1}{\cS(\hbar\partial_z)}x.
\end{equation}
The last equality holds because $\partial_z^kz=0$ for $k\ge2$.

Adding to $y$ a summand proportional to $x$ does not change topological recursion. It follows that the spectral curve data for the same TR differentials~\eqref{eq:linearHodge} can be written in the form
\begin{equation}\label{eq:Lambert-logz}
x=z-\log(z),\quad y=\log z.
\end{equation}
This provides a different $x-y$ duality formula for the same differentials. Note that the logarithmic singularities of~$x$ and $y$ are not LogTR-vital in this case anymore, therefore $\{{}^{\log{}}\omega^{\vee,(g)}_n\}=\{\omega^{\vee,(g)}_n\}$ and $\{{}^{\log{}}\omega^{(g)}_m\}=\{\omega^{(g)}_m\}$, and Proposition~\ref{prop:DualTRTrivial} provides a correct answer with undeformed $\hat x=x$ in this case. The obtained expression for the differentials~\eqref{eq:linearHodge} have been known before. It is equivalent to a special case of a closed formula of~\cite{BDKS1} corresponding to the specific choice~$\psi(y)=y$, $y(z)=z$ in \emph{op. cit}. An observation that these expressions can be treated in the framework of extended $x-y$ duality is due to Hock~\cite{hock2023xy}.

\subsection{\texorpdfstring{$r$}{r}-spin vs \texorpdfstring{$r$}{r}-orbifold Hurwitz numbers}

Generalizing an example with the Lambert spectral curve, consider two families of the topological recursions with the following spectral curves:
\begin{align}
x&=P(z)-\log z,\qquad y=z,\label{eq:r-spin}\\
x&=P(z)-\log z,\qquad y=\log z,\label{eq:r-orbifold}
\end{align}
where $P$ is a rational function. In the case $P(z)=z^r$ these families describe enumeration of the so called $r$-spin and $r$-orbifold Hurwitz numbers, respectively \cite{Orbifold,SSZ}. In opposite to the case of the Lambert curve, where \eqref{eq:Lambert-z} and \eqref{eq:Lambert-logz} lead to the the same answers, topological recursion for these families, in general, result to different systems of differentials.

The function $x$ for these two families is the same and possess logarithmic singularities at $z=0$ and $z=\infty$. However, a LogTR-vital singularity appears only in one of two cases. Namely, if $y=z$, then a LogTR-vital singularity of~$x$ is $z=0$ only; if $y=\log z$, then none of the logarithmic singularities are LogTR-vital. According to this, we conclude that the TR differentials for the spectral curves~\eqref{eq:r-spin} and~\eqref{eq:r-orbifold}, respectively, are given by~\eqref{eq:main} with $\hat x$ defined by
\begin{align}
y&=z,&\hat x&=P(z)-\frac{1}{\cS(\hbar\partial_y)}\log z,
\\y&=\log z,&\hat x&=x,
\end{align}
and $w_i,\bar w_i$ defined by~\eqref{eq:wiexplicit}.

Moreover, $y$ never has a LogTR-vital singularity, therefore $\{{}^{\log{}}\omega^{(g)}_m\}=\{\omega^{(g)}_m\}$.

\subsection{Triple Hodge integrals}
Another natural generalization of the Lambert spectral curve data is the one defined by
\begin{equation}\label{eq:3HSpectralCurve}
dx=\frac{z\,dz}{(1+\alpha z)(1+\beta z)(1+\gamma z)},\qquad
dy=\frac{(1+b z)dz}{(1+\alpha z)(1+\beta z)(1+\gamma z)},
\end{equation}
where $z$ is an affine coordinate on~$\Sigma=\C P^1$, and $\alpha,\beta,\gamma,b$ are complex numbers. More explicitly, if $\alpha$, $\beta$, and $\gamma$ are different, we have
\begin{equation}
x=\frac{1}{p}\log(1+\alpha z)+\frac{1}{q}\log(1+\beta z)+\frac{1}{r}\log(1+\gamma z),
\end{equation}
where the parameters~$p,q,r$ are related to $\alpha,\beta,\gamma$ by
\begin{equation}\label{eq:pqr}
\begin{aligned}
p&=-(\alpha-\beta)(\alpha-\gamma),\\
q&=-(\beta-\gamma)(\beta-\alpha),\\
r&=-(\gamma-\alpha)(\gamma-\beta).
\end{aligned}
\end{equation}
These parameters automatically satisfy the Calabi--Yau condition $\frac1p+\frac1q+\frac1r=0$.

Both $dx$ and $dy$ are meromorphic and have three poles (the same for $dx$ and $dy$) and one zero (different for~$x$ and~$y$). Recall formula~\eqref{eq:xi} which defines the functions $\xi_k(z)$,~$k\ge0$. For the specific case which we currently consider we have
\begin{equation}\label{eq:xi3H}
\xi_k(z)=(-\partial_x)^{k}\left(-\frac1z-b\right).
\end{equation}
Then the differentials of TR with this spectral curve data are given by triple Hodge integrals~\cite{BM,ZhouTV,Chen}
\begin{gather}\label{eq:omega3H}
\omega^{(g)}_n(z_{\set{n}})=\sum_{k_1+\dots+k_n\le3g-3+n}
\int\limits_{\overline{\mathcal M}_{g,n}}\Lambda(p)\Lambda(q)\Lambda(r)\,\psi_1^{k_1}\dots\psi_n^{k_n}\;d\xi_{k_1}(z_1)\dots d\xi_{k_n}(z_n),\\
\Lambda(t)=1+t\lambda_1+t^2\lambda_2+\dots+t^g\lambda_g.
\end{gather}
Sometimes it is convenient to consider such affine coordinate that $dx$ has a pole at infinity; this type of parametrization is quite common in the literature.


We wish now to write a closed formula for these differentials using the $x-y$ duality relation. Observe that both sides of~\eqref{eq:omega3H} are independent of~$b$. Indeed, changing the value of~$b$ in the definition of~$y$ changes this function by a summand proportional to $x$, which preserves TR differentials on the left hand side of~\eqref{eq:omega3H}. The right hand side is also independent of~$b$ since a choice of~$b$ affects the constant term of $\xi_0$ only.

For a general value of~$b$ the differential $dy$ has a zero and Proposition~\ref{prop:DualTRTrivial} is not applicable. But there are three exceptional values of the parameter~$b$ for which the zero in the numerator of the expression for $dy$ (Eq.~\eqref{eq:3HSpectralCurve}) gets canceled by one of the three factors in the denominator. In these cases $dy$ has no zeros and (a version of)~\eqref{eq:main} provides a closed formula for the differentials of logarithmic topological recursion. In fact, we obtain \emph{three non-equivalent expressions} in this way. Consider one of them which corresponds to the choice $b=\gamma$. That is, we consider the spectral curve data
\begin{equation}
dx=\frac{z\,dz}{(1+\alpha z)(1+\beta z)(1+\gamma z)},\qquad
dy=\frac{dz}{(1+\alpha z)(1+\beta z)},
\end{equation}
which is equivalent to
\begin{equation}\label{eq:xy-3H}
x=\tfrac{1}{q}\log\bigl(\tfrac{1+\beta z}{1+\alpha z}\bigr)+\tfrac{1}{r}\log\bigl(\tfrac{1+\gamma z}{1+\alpha z}\bigr),\qquad
y=\tfrac{1}{\beta-\alpha}\log\bigl(\tfrac{1+\beta z}{1+\alpha z}\bigr).
\end{equation}
\begin{remark}
In this parametrization it is easy to see that if two of the three poles of~$dx$ are coinciding, then this situation is reduced to the case of linear Hodge integrals considered in Section \ref{SHodge}, by a suitable change of the coordinate~$z$. If all three poles are coinciding, then the spectral curve reduces to the Airy spectral curve $(z^2,z)$ and the differentials of TR describe the integrals of just psi classes over moduli spaces.
\end{remark}

We observe that among three logarithmic singularities of~$x$ only one is LogTR-vital, namely, that one at the point $z=-1/\gamma$. Taking $\frac{1+\beta z}{1+\alpha z}=e^{(\beta-\alpha)y}$ as an affine coordinate of Proposition~\ref{prop:DualTRTrivial} we see that an expression~\eqref{eq:main} for logarithmic TR differentials is valid in this case with
\begin{equation}\label{eq:w-3H}
\begin{aligned}
w_i&=e^{(\beta-\alpha)(y_i+u_i\hbar/2)}=e^{(\beta-\alpha)u_i\hbar/2}\tfrac{1+\beta z_i}{1+\alpha z_i},
\\\bar w_i&=e^{(\beta-\alpha)(y_i-u_i\hbar/2)}=e^{-(\beta-\alpha)u_i\hbar/2}\tfrac{1+\beta z_i}{1+\alpha z_i},
\\ \hat x&=\tfrac{1}{q}\log\bigl(\tfrac{1+\beta z}{1+\alpha z}\bigr)+
\frac{1}{r\cS(r\hbar\partial_y)}
\log\bigl(\tfrac{1+\gamma z}{1+\alpha z}\bigr)=
\frac{1}{\cS(r\hbar\partial_y)}x,
\end{aligned}
\end{equation}
where $r$ is the last entry of the triple $(p,q,r)$ defined by~\eqref{eq:pqr}. The last equality holds because $\partial_y^r\log\bigl(\tfrac{1+\beta z}{1+\alpha z}\bigr)=0$ for $r\ge2$.

Moreover, $y$ does not have LogTR-vital singularities, thus $\{{}^{\log{}}\omega^{(g)}_m\}=\{\omega^{(g)}_m\}$.

\subsection{Insertion of \texorpdfstring{$\kappa$}{k} classes}\label{sec:kappaInsertion}

The form~$dx$ of the spectral curve data~\eqref{eq:3HSpectralCurve} is just a general rational $1$-form with a unique zero and three poles, while the respective choice of~$dy$ in~~\eqref{eq:3HSpectralCurve} is rather special for a given $dx$. What can we say about a more general choice of $y$? In fact, in order to have a well defined topological recursion, it is sufficient to assume that $y(z)$ is just an arbitrary power series expanded at the origin such that $y'(0)\ne0$. For any such series define the constants $q_i,s_i$, $i=1,2,\dots$ by the relations
\begin{gather}\label{eq:parametersinkappa}
q_k=-\res_{z=0}\xi_k(z)dx \int^z y\,dx
  =\res_{z=0}d\xi_{k+1}(z)\int^z y\,dx
  =-\res_{z=0} \xi_{k+1}(z)\,y\,dx,
\\ \notag 1+\frac{q_1}{q_0}t+\frac{q_2 }{q_0}t^2+\dots =e^{-\sum _{i=1}^{\infty } s_i t^i},
\end{gather}
where $\xi_k$ are given by~\eqref{eq:xi}. Then the differentials of topological recursion for this choice of~$y$ and for $x$ defined by~\eqref{eq:3HSpectralCurve} are given by~\cite{Eynard-intersections}
\begin{equation}\label{eq:omega3Hkappa}
\omega^{(g)}_n=q_0^{2-2g-n}\!\!\!\sum_{k_1+\dots+k_n\le3g-3+n}
\int\limits_{\overline{\mathcal M}_{g,n}}e^{\sum_{j=1}^\infty s_j\kappa_j}
\Lambda(p)\Lambda(q)\Lambda(r)\,\psi_1^{k_1}\dots\psi_n^{k_n}\;d\xi_{k_1}(z_1)\dots d\xi_{k_n}(z_n).
\end{equation}

Of course, for a general choice of the coefficients $q_k$ the differential $dy$ is not meromorphic on the spectral curve, and Proposition~\ref{prop:DualTRTrivial} is not applicable. However, one can consider a special case when $dy$ is of the form
\begin{equation}
dy=\frac{dz}{(1+\tilde{\alpha} z)(1+\tilde{\beta} z)}
\end{equation}
for arbitrary parameters $\tilde{\alpha}$ and $\tilde{\beta}$. Then it has no zeros and Proposition~\ref{prop:DualTRTrivial} can be applied. But, if either $\tilde{\alpha}$ or  $\tilde{\beta}$ is not equal to one of the parameters $\{\alpha,\beta,\gamma\}$, then $y$ has LogTR-vital singularities and $\{{}^{\log{}}\omega^{(g)}_m\}\neq\{\omega^{(g)}_m\}$.

As an example, we could consider the following spectral curve data:
\begin{equation}\label{eq:R=z+1}
x=z-\log(z),\quad y=\log(z+1).
\end{equation}
The differentials ${}^{\log{}}\omega^{(g)}_n$ computed by~\eqref{eq:main} for this spectral curve data, along with the ``expected'' poles at the point $z=1$ (which is a unique zero of~$dx$), possess also (for~$n=1$) an additional pole at $z=-1$ which is a LogTR-vital singularity of $dy$ for this spectral curve data. In order to obtain differentials of the original topological recursion for this spectral curve one should apply further the recomputation procedure considered below in Section~\ref{sec:recomputation}. Enumerative meaning of the topological recursion differentials in this case is also clear --- they describe linear Hodge integrals with kappa classes inserted, as explained above.

For a slightly different choice of $y$,
\begin{equation}\label{eq:nontriv}
x=z-\log(z),\quad y=\log(z+1)+z^3,
\end{equation}
Proposition~\ref{prop:DualTRTrivial} is not applicable. This curve provides an example of the situation when both $\{{}^{\log{}}\omega^{\vee,(g)}_n\}\neq\{\omega^{\vee,(g)}_n\}$ and $\{{}^{\log{}}\omega^{(g)}_m\}\neq\{\omega^{(g)}_m\}$, moreover, all these four systems of differentials are non-trivial, and
$\{\omega^{(g)}_m\}$ describe linear Hodge integrals with certain insertion of kappa classes.

Note that in all other examples of this section, Eq.~\eqref{eq:main} resulted in differentials of the corresponding original topological recursions, that is,
${}^{\log{}}\omega^{(g)}_n=\omega^{(g)}_n$ for all $g$ and $n$.

\subsection{Insertion of \texorpdfstring{$\Theta$}{Theta} classes}

Consider yet another peculiar example of topological recursion with one ramification point. It is defined by
\begin{equation}
dx=\frac{z\,dz}{(1+\alpha z)(1+\beta z)(1+\gamma z)},\qquad
y=\frac{1}{z}.
\end{equation}

The function~$y$ has a (simple) pole at $z=0$ where~$dx$ vanishes, which means that 
topological recursion becomes irregular in this case. 
Its differentials are given by~\cite{H3Theta}
\begin{equation}\label{eq:omega3HTheta}
\omega^{(g)}_n=\sum_{k_1+\dots+k_n\le3g-3+n}
\int\limits_{\overline{\mathcal M}_{g,n}}\Theta_{g,n}
\Lambda(p)\Lambda(q)\Lambda(r)\,\psi_1^{k_1}\dots\psi_n^{k_n}\;d\xi_{k_1}(z_1)\dots d\xi_{k_n}(z_n),
\end{equation}
where the class  $\Theta_{g,n}$ that was introduced initially by Norbury \cite{NorbTheta} is proved to be expressible as a polynomial in~$\kappa$ classes \cite{KN,CGG}. Explicitly, it is given by
\begin{equation}
\Theta_{g,n}=[t^{2g-2+n}]e^{\sum_{j=1}^\infty s_j\kappa_jt^j},
\quad \text{where}\quad e^{-\sum_{j=1}^\infty s_jt^j}=\sum_{k=0}^\infty(-1)^k(2k+1)!t^k.
\end{equation}

Relations of the previous section show that this irregular topological recursion can be obtained as a limit from a regular one of the previous section with a certain specialization of parameters $s_i$ from \eqref{eq:parametersinkappa} and with a suitable rescaling of them, see~\cite{KN}.

Then, since $y$ is unramified in the family before the limit, we conclude that Proposition~\ref{prop:DualTRTrivial} can be applied. All three poles of $dx$ are LogTR-vital this time, and we obtain that Eq.~\eqref{eq:main} after taking the limit holds for the differentials~\eqref{eq:omega3HTheta} with
\begin{equation}
\hat x=\frac{1}{p\cS(p\hbar\partial_y)}\log\bigl(\tfrac{1+\alpha z}{z}\bigr)
+\frac{1}{q\cS(q\hbar\partial_y)}\log\bigl(\tfrac{1+\beta z}{z}\bigr)
+\frac{1}{r\cS(r\hbar\partial_y)}\log\bigl(\tfrac{1+\gamma z}{z}\bigr).
\end{equation}
Again, $y$ does not have LogTR-vital singularities, thus $\{{}^{\log{}}\omega^{(g)}_m\}=\{\omega^{(g)}_m\}$.

This example is interesting because in general all three logarithmic summands use different and non-integer scaling parameters of~$\hbar$-deformations, and this cannot be avoided by a change of coordinates.

Of course, this example can be further extended by deforming $y$ and an appropriate insertion of $\kappa$-classes.

\begin{remark}
	We use here the notion of irregular topological recursion meaning that there is a way to handle with situations when $x$ or $y$ is not regular in the zeroes of $dy$ or $dx$ correspondingly. We carefully treat these situations in the subsequent paper \cite{ABDKS7}.
\end{remark}

\subsection{Topological recursion for Bousquet-M\'elou--Schaeffer numbers}

Topological recursion for the Bousquet-M\'elou--Schaeffer (BMS) numbers provides yet another manifestative example clarifying details of applicability of Proposition~\ref{prop:DualTRTrivial}. Following~\cite{bychkov2021topological}, we represent the corresponding spectral curve in the form
\begin{equation}
X=\frac{z}{(z+1)^m},\quad y=z,\quad \omega^{(0)}_1=z\frac{dX}{X}.
\end{equation}

In order to represent this spectral curve data in a more traditional  form we should choose $x$ of the topological recursion spectral curve data as some function of $X$ and to adjust $y$ accordingly in order to obtain $\omega^{(0)}_1$ in its conventional form $-y\,dx$. One of the ways to do that is to set $dx=-\frac{dX}{X}$, that is,
\begin{equation}\label{eq:BMShcurve}
x=m \log(1+z)-\log(z),\quad y=z.
\end{equation}
For this data we have that $dx=\frac{(m-1)z-1}{z(z+1)}dz$ has a unique (nondegenerate) zero point $z=\frac{1}{m-1}$ that might contribute to topological recursion. The differentials of this topological recursion can be expressed explicitly by~\eqref{eq:main} with
\begin{equation}\label{eq:BMSh_xhat1}
\hat x=\frac{m}{\cS(m^{-1}\hbar\partial_y)}\log(1+z)+\frac{1}{\cS(\hbar\partial_y)}\log(z).
\end{equation}

The point is that it is \emph{not} the topological recursion we need. The correct form of topological recursion for the BMS numbers takes also into account the ramification point $z=-1$ of the function~$X$. In order to do that, the paper~\cite{BDS} suggests to represent $\omega^{(0)}_1$ in the form
\begin{equation}
\omega^{(0)}_1=-zX\,d(X^{-1})
\end{equation}
and to consider $x=X^{-1}=\frac{(z+1)^m}{z}$ and $y=zX=\frac{z^2}{(z+1)^m}$. The function~$x$ chosen in this way has two critical points, a nondegenerate one at $z=\frac{1}{m-1}$ and a degenerate one at $z=-1$, and the correct form of topological recursion related to enumeration of BMS numbers takes into account both these points. This recursion is not just degenerate but also irregular since the $y$ function has a pole at the corresponding degenerate critical point $z=-1$ of the $x$ function. The initial data for this recursion is quite complicated and Proposition~\ref{prop:DualTRTrivial} is not applicable.

Instead, we suggest to consider the original form of the spectral curve but with the factor $(z+1)^m$ replaced by a generic polynomial of degree~$m$. Thus, we consider spectral curve data of the form
\begin{equation}
x=\log P_\epsilon(z)-\log(z),\quad y=z,
\end{equation}
where $P_\epsilon$ is a family of polynomials depending on $\epsilon$ and converging to $P_0=(z+1)^m$ as $\epsilon\to0$. This limit procedure is equivalent to the one used in~\cite[Section 2]{BDS} (note that the argument in op.~cit. used a statement on taking limits in topological recursion that was not available in the literature at that moment; this missing step is now performed in~\cite{limits}), and thus this spectral curve data does give the BMS numbers at $\epsilon\to 0$.

Then, we apply~\eqref{eq:main} to the spectral curves of this family and take the limit of the obtained differentials as $\epsilon\to 0$. The resulting expression for the limiting differentials is given by~\eqref{eq:main} with
\begin{equation}\label{eq:BMSh_xhat2}
\hat x=\frac{1}{\cS(\hbar \partial_y)}\log ((z+1)^m)-\frac{1}{\cS(\hbar \partial_y)}\log(z)
=\frac{1}{\cS(\hbar \partial_y)}x(z).
\end{equation}

Thus, we obtained two expressions of the form~\eqref{eq:main} for the same spectral curve~\eqref{eq:BMShcurve} but for different $\hbar$-deformations of the function~$x$ given by~\eqref{eq:BMSh_xhat1} and~\eqref{eq:BMSh_xhat2}, respectively. The first one provides differentials of topological recursion with the only ramification point $z=\frac{1}{m-1}$ taken into account, and the second one, which is actually related to enumeration of the BMS numbers, accounts also correctly the contribution of the point $z=-1$.

\subsection{The Hock conjecture} Hock~\cite[Section 3.2]{hock2023xy} considered the following special situation (we adapt his proposal to the notation and conventions of the present paper):

\begin{conjecture}[Hock] Let $\Sigma=\C\mathrm{P}^1$, $y(z)=z$ be a global affine coordinate, and $x=\log R(z)+Az$, where $R(z)$ is a rational function and $A\in\C$ (Case 1). Alternatively, let $\Sigma=\C\mathrm{P}^1$, $y(z)=\log z$, where $z$ is a global affine coordinate, and $x=\log R(z)$, where $R(z)$ is again a rational function (Case 2).
	
Assume that all zeros of $dx$ are simple. Then the correlation differentials $\{\omega^{(g)}_n\}$ can be obtained by the $x-y$ swap relation from the dual system of differentials formally defined as
\begin{align}
	\sum_{g=0}^\infty \hbar^{2g}\omega^{\vee,(g)}_{1} = - dy(z) \frac{1}{\cS(\hbar \partial_{y(z)})} x(z),
\end{align}
and $\omega^{\vee,(g)}_n = 0$ for $g\geq 0$, $n\geq 2$, $2g-2+n>0$.
\end{conjecture}

We have the following formal corollary of Theorem~\ref{thm:LogXYSwap} (or rather Proposition~\ref{prop:DualTRTrivial}):

\begin{corollary}\label{corHock} The conjecture of Hock holds under extra assumptions:
\begin{enumerate}
	\item All finite zeros and poles of $R(z)$are simple (in Case 1);
	\item $R(z)$ has either a zero or a pole at both $z=0$ and $z=\infty$ (not necessarily simple ones), and all other zeros and poles of $R(z)$ are simple (in Case 2).
\end{enumerate}
\end{corollary}

\begin{proof}
If all finite zeros and poles of $R(z)$ are simple then $\log R$ is a sum of logarithms of linear functions taken with the coefficients $\pm1$. It follows that Proposition~\ref{prop:DualTRTrivial} can be applied with
\begin{equation}
\hat x=\begin{cases}
	\frac{1}{\cS(\hbar\partial_y)}\log R(z)+Az & \text{in Case 1}; \\
	\frac{1}{\cS(\hbar\partial_y)}\log R(z) & \text{in Case 2}.
\end{cases}
\end{equation}
Since $\partial_y^rz=0$ for $r\ge2$, in both cases we have
\begin{equation}
	\hat x = \frac{1}{\cS(\hbar\partial_y)}x.
\end{equation}
In Case 1 it remains to note that the function $y(z)=z$ has no logarithmic singularities implying that the LogTR differentials produced by~\eqref{eq:main} in this case coincide with TR differentials. In Case 2 $y(z)=\log z$ has logarithmic singularities at $z=0$ and $z=\infty$, but the second assumption implies that they are not LogTR-vital, hence the LogTR differentials produced by~\eqref{eq:main}  coincide with TR differentials in this case as well, cf.~Remark~\ref{rem:Appl-dual-Log-Ord}.
\end{proof}

%
%

\begin{remark} We see in the proof that the fact that the same $\hbar$-correction can be applied to each summand of~$x$ is just a coincidence and this does not hold for more complicated logarithmic singularities. For example, if in Case 1 $R(z)$ has a finite multiple pole or a multiple zero, then the conjecture of Hock does not hold. Indeed, in this case the multiplicity is the residue (or the opposite residue for a pole) of $dy^\vee$ and it enters the formula for the necessary correction~\eqref{eq:PrincipalPartsLog}.
\end{remark}







\subsection{Family~II of generalized double Hurwitz numbers} \label{sec:fam2}

In~\cite{bychkov2021topological} two families (Family~I and Family~II) of KP tau functions of hypergeometric type (whose expansion coefficients are various weighted double Hurwitz numbers) were considered, and it was proved that the respective $n$-point functions (for which \cite[Proposition~2.2]{bychkov2021topological} provides an explicit formula) satisfy topological recursion on specific spectral curves.

In~\cite[Section~7.1]{alexandrov2022universal} it was shown that ordinary $x$--$y$ swap relation provides an alternative way of obtaining the explicit formula for the $n$-point functions for a part of Family~I.

It turns out that the $x$--$y$ swap relation for the logarithmic topological recursion provides an alternative way for obtaining explicit formulas for the $n$-point functions for the case of Family~II (the whole family without any additional restrictions).

Indeed, recall that for the Family~II we have
\begin{align}
	x(z)&= -\log z + \psi(y(z)) =  -\log z + \alpha\left(\dfrac{R_1(z)}{R_2(z)} + \log\left(\dfrac{R_3(z)}{R_4(z)}\right)\right),\\
	y(z)&=\dfrac{R_1(z)}{R_2(z)} + \log\left(\dfrac{R_3(z)}{R_4(z)}\right),
\end{align}
where $\alpha\in\mathbb{C}\setminus \{0\}$, and $R_i$ are polynomials with simple roots such that $y$ is not identically zero but is vanishing at $z=0$; $\psi(\theta)=\alpha\theta$ is one of the functions specifying the given KP tau function of hypergeometric type, see~\cite{BDKS1,bychkov2021topological}.

Let us denote by $\{\omega^{(g)}_n\}$ the $n$-point differentials produced by TR on the spectral curve  $(\mathbb{C}P^1,x,y)$. \cite[Proposition~2.2]{bychkov2021topological} provides an explicit formula for $W^{(g)}_n=\omega^{(g)}_n/\prod_{i=1}^n(-dx_i)$. Note that for the function $X$ which is used in \cite{bychkov2021topological} we have $X=e^{-x}$ (recall that in the present paper the convention is that $\omega^{(0)}_1=-ydx$ which then coincides with the convention $\omega^{(0)}_1=ydX/X$ from \cite{bychkov2021topological}).

Note that if one sets
\begin{align}
	\tilde y(z)&:=y(z)-\dfrac{x(z)}{\alpha}=\dfrac{\log z}{\alpha},
\end{align}
then the $n$-point differentials $\{\tilde\omega^{(g)}_n\}$ produced by TR on the spectral curve $(\mathbb{C}P^1,x,\tilde y)$ 
exactly coincide with $\{\omega^{(g)}_n\}$ for all $(g,n)$ except $(g,n)=(0,1)$.

\begin{remark}
Family~II includes the cases \eqref{eq:Lambert-logz}, \eqref{eq:r-orbifold} considered above.
\end{remark}

Let us put $\alpha=1$ for brevity in what follows; it is easy to recover it.
Also we would like to avoid LogTR-vital singularities for this spectral curve. Namely, we assume that $dx$ has poles at the points $z=0$ and $z=\infty$ which are poles of $dx^\vee=d\tilde y=dz/z$. This requirement is implied by the generality conditions imposed on the polynomials $R_1,\dots,R_4$. Thus, if one denotes by $\{{}^{\log{}}\omega^{(g)}_n\}$ the $n$-point differentials produced by LogTR on the spectral curve $(\mathbb{C}P^1,x,\tilde y)$, then they coincide with $\{\tilde\omega^{(g)}_n\}$ for all $(g,n)$ and with $\{\omega^{(g)}_n\}$ for all $(g,n)$ except $(g,n)=(0,1)$.

Now let us apply Theorem~\ref{thm:LogXYSwap}. Denote by $\{{}^{\log{}}\omega^{\vee,(g)}_n\}$ the $n$-point differentials produced by LogTR on the spectral curve $(\mathbb{C}P^1,x^\vee,y^\vee)$, where $x^\vee=\tilde y$ and $y^\vee=x$. Theorem~\ref{thm:LogXYSwap} implies that $\{{}^{\log{}}\omega^{(g)}_n\}$ can be expressed in terms of $\{{}^{\log{}}\omega^{\vee,(g)}_n\}$ via formula~\eqref{eq:MainFormulaInverse}.

Denote
\begin{align}
	\hat y(z)&:=\dfrac{R_1(z)}{R_2(z)} + \frac{1}{\cS(\hbar z\partial_z)}\log\left(\dfrac{R_3(z)}{R_4(z)}\right),\\
	\hat x(z)&:=-\log z+\dfrac{R_1(z)}{R_2(z)} + \frac{1}{\cS(\hbar z\partial_z)}\log\left(\dfrac{R_3(z)}{R_4(z)}\right)=  -\log z +\hat y(z).
\end{align}
Then we note that Definition~\ref{def:log-proj} implies that for $n>1$ and $(g,n)\neq (0,2)$ we have ${}^{\log{}}\omega^{\vee,(g)}_n =0$, since $x^\vee$ does not have any critical points; moreover, it implies that for $g>0$
\begin{equation}\label{eq:fam2logpoles}
	\dfrac{{}^{\log{}}\omega^{\vee,(g)}_1(z)}{dx^\vee(z)} = -[\hbar^{2g}]\hat x(z)= -[\hbar^{2g}]\hat y(z)= -[\hbar^{2g}]\frac{1}{\cS(\hbar z\partial_{z})} \log\left(\dfrac{R_3(z)}{R_4(z)}\right).
\end{equation}

Now let us substitute everything into formula~\eqref{eq:MainFormulaInverse}, for $(g,n)\neq(0,1)$:
\begin{align}\label{eq:WF2formula}
	&W^{(g)}_n (z_{\llbracket n\rrbracket})=\frac{\omega^{(g)}_n (z_{\llbracket n\rrbracket})}{\prod_{i=1}^n (-dx_i)} = (-1)^n \frac{\tilde\omega^{(g)}_n (z_{\llbracket n\rrbracket})}{\prod_{i=1}^n dx_i}  = (-1)^n\frac{{}^{\log{}}\omega^{(g)}_n (z_{\llbracket n\rrbracket})}{\prod_{i=1}^n dx_i} \\ \notag
	&\mathop{=}^{\eqref{eq:MainFormulaInverse}}	
	\coeff \hbar {2g}\hbar^{2} \sum_{\Gamma} \frac{1}{|\mathrm{Aut}(\Gamma)|} \prod_{i=1}^n \hbar^{-2}
	\sum_{k_i=0}^\infty  \partial_{y^\vee_i}^{k_i} [u_i^{k_i}] \frac{dx^\vee_i}{dy^\vee_i}
	\\ \notag
	&\phantom{==}
	\frac{1}{ u_i} e^{ u_i \cS(\hbar u_i \partial_{x^\vee_i}) \sum_{\tilde g=0}^\infty \hbar^{2\tilde g} \frac{{}^{\log{}}\omega^{\vee,(\tilde g)}_{1} (z_i)}{dx^\vee_i}-u_i \frac{{}^{\log{}}\omega^{\vee,(0)}_{1} (z_i)}{dx^\vee_i}}
	\\ \notag
	&\phantom{==}
	\prod_{e\in E(\Gamma)}\hbar^{-2}\prod_{j=1}^{|e|\geq 2}\restr{(\tilde u_j, \tilde x^\vee_j)}{ (u_{e(j)},x^\vee_{e(j)})} \hbar^2 \tilde  u_j \cS(\hbar \tilde u_j  \partial_{\tilde x^\vee_j}) \sum_{\tilde g=0}^\infty \hbar^{2\tilde g}\, \frac{{}^{\log{}}\tilde\omega^{\vee,(\tilde g)}_{|e|}(\tilde z_{\llbracket |e|\rrbracket})}{\prod_{j=1}^{|e|} d\tilde x^\vee_j}\\\notag	
	&=	
	\coeff \hbar {2g-2+n} \prod_{i=1}^n
	\sum_{k_i=0}^\infty  \partial_{x_i}^{k_i} [u_i^{k_i}] \left(z_i\frac{dx_i}{dz_i}\right)^{-1}
	\frac{e^{ -u_i \cS(\hbar u_i z_i\partial_{z_i})\hat x(z_i)+u_i x(z_i)}}{\hbar u_i}
	\\ \notag  &\phantom{==}
	e^{\frac12
		\restr{\tilde z_1}{z_i} \restr{\tilde z_2}{z_i} \hbar^2 u_i^2\cS(u_i\hbar \tilde z_1\partial_{\tilde z_1})\cS(u_i\hbar\tilde z_2\partial_{\tilde z_2})
		\left(\frac{\tilde z_1 \tilde z_2}{(\tilde z_1-\tilde z_2)^2}-\frac{1}{(\log \tilde z_1-\log \tilde z_2)^2}\right)} \\ \notag
		&\phantom{==}
		\sum_{\gamma} \prod_{(v_k,v_\ell)\in E(\gamma)}\left(e^{
			\hbar^2 u_k u_\ell \cS(u_k\hbar z_k\partial_{z_k})\cS(u_\ell\hbar z_\ell\partial_{z_\ell})
			\frac{z_k z_\ell}{( z_k- z_\ell)^2}}-1\right)
\\\notag	&=						
    \coeff \hbar {2g-2+n} \prod_{i=1}^n
			\sum_{k_i=0}^\infty  (-\partial_{x_i})^{k_i} [u_i^{k_i}] \left(-z_i\frac{dx_i}{dz_i}\right)^{-1}
			\frac{e^{ u_i \cS(\hbar u_i z_i\partial_{z_i})\hat y(z_i)-u_i y(z_i)}}{u_i\hbar\cS(u_i\hbar)}
	\\ \notag		&\phantom{==}
			\sum_{\gamma}  \prod_{(v_k,v_\ell)\in E(\gamma)}\left(e^{
			\hbar^2 u_k u_\ell \cS(u_k\hbar z_k\partial_{z_k})\cS(u_\ell\hbar z_\ell\partial_{z_\ell})
			\frac{z_k z_\ell}{( z_k- z_\ell)^2}}-1\right).
\end{align}
In the last and the third lines from the bottom the sum goes over connected simple graphs $\gamma$ (as opposed to multigraphs $\Gamma$), and the last equality follows from the fact that
\begin{equation}
	e^{\frac12
		\restr{\tilde z_1}{z_i} \restr{\tilde z_2}{z_i} \hbar^2 u_i^2\cS(u_i\hbar \tilde z_1\partial_{\tilde z_1})\cS(u_i\hbar\tilde z_2\partial_{\tilde z_2})
		\left(\frac{\tilde z_1 \tilde z_2}{(\tilde z_1-\tilde z_2)^2}-\frac{1}{(\log \tilde z_1-\log \tilde z_2)^2}\right)}
\end{equation}
does not, in fact, depend on $z_i$ and is equal precisely to $1/\cS(u_i\hbar)$. We have also replaced $u_i$ with $-u_i$ in the last equality. Also note that in the equality in the second line, apart from formula~\eqref{eq:MainFormulaInverse} itself, we have also used the formula for the first Betti number of a multigraph: $g(\Gamma)=|L|+1 -|V|-|E|$, where $|V|$ is the number of vertices, $|E|$ is the number of multiedges, and $|L|$ is the total number of all ``legs'' of all multiedges.

Note that the expression that we have obtained in~\eqref{eq:WF2formula} exactly coincides with the formula from~\cite[Proposition~2.2]{bychkov2021topological} specified to the case of Family~II (cf. \cite[Section~4.3]{bychkov2021topological}). Indeed, note that for $D=X\partial_X$ and $Q=(z/X)dX/dz$ from~\cite{bychkov2021topological} we have $D=-\partial_x$ and $Q=-z dx/dz$.

Thus we get a new proof of TR for Family~II (although some of the ideas behind the original proof of TR for Family~II and behind the proof of Theorem~\ref{thm:LogXYSwap} which we used here are similar).


\begin{remark}
	It is important to note that from the above reasoning one clearly sees that in this case the difference~\eqref{eq:fam2logpoles} between the 1-point differentials given by LogTR and by original TR on the trivial side precisely corresponds to the difference between the $\hbar$-deformed function $\hat y$ and the regular $y$ in \cite{bychkov2021topological}, and only due to this coincidence the formulas end up coinciding.
\end{remark}

\begin{remark}
	Note that Family~II from~\cite{bychkov2021topological} contains e.g. the Ooguri--Vafa partition function of colored HOMFLY polynomials for torus knots. That is to say, some members of this family have independent significance.
\end{remark}
\subsection{KP integrability}

In \cite{ABDKS3} we introduced a concept of KP integrability as
a property of a system of symmetric meromorphic differentials and we proved, in particular, that the $x-y$ swap relation respects this property. The differentials of LogTR and its dual are connected to each other by exactly this relation, therefore, from \cite[Theorem 2.7]{ABDKS3} we immediately have
\begin{proposition}
The system of differentials $\{{}^{\log{}}\omega^{(g)}_m\}$  of LogTR is KP integrable if and only if the dual system $\{{}^{\log{}}\omega^{\vee,(g)}_n\}$ is KP integrable.
\end{proposition}

If $dy$ has no zeros and the dual TR is trivial, then the system $\{{}^{\log{}}\omega^{\vee,(g)}_n\}$ is KP integrable. Indeed, it immediately follows form the definition that any system of differentials  on ${\mathbb P}^1$ such that for $n>1$ $\omega^{(g)}_n=\delta_{g,0}\delta_{n,2}B$ and $\omega^{(g)}_1$ are arbitrary meromorphic differentials is KP integrable. Therefore, if the dual TR is trivial then the system of differentials $\{{}^{\log{}}\omega^{(g)}_m\}$ is KP integrable. It proves the KP integrability of all enumerative geometry interpretations of the TR considered in this section, which, however, was already established by other methods.

\subsection{A case with nontrivial dual:
	Family~III of generalized Hurwitz numbers}
As mentioned in Section~\ref{sec:fam2}, in~\cite{bychkov2021topological} it was proved that the $n$-point functions of two families of KP tau functions of hypergeometric type (a.k.a. Orlov--Scherbin partition functions), which correspond to generalized double Hurwitz numbers (and cover most examples of such numbers studied in the literature), satisfy the original topological recursion (a new proof of that for Family II is given in Section~\ref{sec:fam2}).

It turns out that there exists another natural family  of such tau functions, whose $n$-point functions do not satisfy the original topological recursion (and there is no known $\hbar$-deformation of the respective $\psi$- and $y$-functions such that the original TR would be satisfied), but they satisfy the LogTR introduced in the present paper. We call this family \emph{Family III} of Orlov--Scherbin partition functions and formulate a theorem that it satisfies the LogTR in the present section.

This example is important since for a generic member of this family both sides of the $x$--$y$ duality are nontrivial, as opposed to all of the examples discussed before (except for the case corresponding to formula~\eqref{eq:nontriv}), and thus Theorem~\ref{thm:LogXYSwap} and Proposition~\ref{prop:DualTRTrivial} cannot be used (or at least they are not enough) for proving LogTR. This illustrates that LogTR can be encountered ``in the wild'' even when there is no obvious relation to the $x$--$y$ duality formula.


\begin{theorem}\label{thm:fam3}
	Let
	\begin{align}
		\psi(\theta) &:= \alpha(e^\theta-1), &\hat\psi(\theta) &:= \cS(\hbar\partial_\theta) \psi(\theta),\\
		y(z) &:= \log R(z), &\hat y(z) &:= \frac{1}{\cS(\hbar z\partial_z)} y(z),\\ \label{eq:fam3xz}
		x(z) &:=-\log z +\psi(y(z)) = -\log z + \alpha\left(R(z) -1\right), & &
	\end{align}
	where $R(z)$ is a non-constant rational function with simple zeros and poles such that $R(0)=1$ and such that zeros of $dx$ are simple and not coinciding with the zeros of $R(z)$.
	
	As discussed in~\cite{bychkov2021topological}, this data defines a 
	KP tau function $Z(p_1,p_2,\dots)$ which is the generating function for certain generalized double Hurwitz numbers, which are then combined into $n$-point differentials $\{\omega^{\hat\psi,\hat y;(g)}_n\}$ (see~\cite[Section~1.2]{bychkov2021topological}).
	
	Then these $\{\omega^{\hat\psi,\hat y;(g)}_n\}$ satisfy LogTR on the spectral curve $(\mathbb{C}P^1,x(z),y(z))$.
\end{theorem}


\begin{remark}
	The assumption on the simplicity of zeros and poles of $R(z)$ can be lifted and the statement of the theorem remains true. This follows from the fact that in this case logarithmic topological recursion behaves nicely with respect to taking the respective limits of spectral curves, which works analogously to taking limits in the case of the original topological recursion, see~\cite{limits} (cf.~\cite[Remark~3.1]{bychkov2021topological}).
\end{remark}
\begin{remark}
	For Family~I of~\cite{bychkov2021topological} all terms in $\psi(\theta)$  except the logarithmic and linear ones had to be $\hbar$-deformed in $\hat\psi$ via the application of $\cS(\hbar\partial_\theta)$, while $y(z)$ is not $\hbar$-deformed. For Family~II, on the other hand, $\psi(\theta)$ is very simple and is not $\hbar$-deformed, while for $y(z)$ the logarithmic terms had to be $\hbar$-deformed via the application of $1/\cS(\hbar z\partial_z)$. Note that in the case of our Family~III both $\psi(\theta)$ and $y(z)$ have to be $\hbar$-deformed with exactly the same operators which were used in Family~I and Family~II for $\psi(\theta)$ and $y(z)$ respectively.
	\end{remark}
	
	\begin{remark}
	Note that Family~III is basically the simplest case when the spectral curve for an Orlov--Scherbin case has LogTR-vital singularities, and indeed it turns out that the respective $n$-point functions for the natural choice of $\hbar$-deformation of $\psi$ and $y$ satisfy the LogTR and not the original TR, while it is completely unclear whether there exists some different choice of $\hbar$-deformations of $\psi$ and $y$ such that the original TR would be satisfied (likely it does not).
\end{remark}
\begin{remark}	
In the special case when $R(z)=1+z,\; \alpha=1,$ Theorem~\ref{thm:fam3} provides an alternative expression for the LogTR differentials associated with the spectral curve \eqref{eq:R=z+1}; this gives rise to a new ELSV-like formula. We postpone the more in-depth discussion of this interesting implication to a subsequent publication~\cite{ABDKS8}.
 \end{remark}

The proof of Theorem~\ref{thm:fam3} 
goes beyond the scope of the present paper; it is proved via a new general technique in our subsequent paper:~\cite[Corollary~6.10]{ABDKS5}. Another proof, via direct analysis in the vein of~\cite{bychkov2021topological}, is given in~\cite{ABDKS8}.

\section{Original vs logarithmic TR} \label{sec:ordinary-vs-log}

The goal of this section is to recall the recursive procedure of the computation of the $n$-point differentials in the case of the original TR and to give a similar presentation for LogTR. Then we specialize~\cite[Equation~2.7]{BS-blobbed} to the case of LogTR thus connecting $\{\omega^{(g)}_m\}$ and $\{{}^{\log{}}\omega^{(g)}_m\}$.

\subsection{Recursive computations} Let us recall the standard algorithm for the recursive computation of $\{\omega^{(g)}_m\}$. It can be used even if $x$ and $y$ have logarithmic singularities, these points are just ignored and the resulting meromorphic differentials are regular at these points.

Assume, as before, that all zeros $\{p_1,\dots,p_N\}$ of $dx$ are simple. Let $\sigma_i$ be the deck transformation with respect to $x$ near $p_i$. As always, we set $\omega^{(0)}_1 \coloneqq -ydx$ and $\omega^{(0)}_2\coloneqq B$. In the case of original TR the correlation differentials $\omega^{(g)}_{m}$, $g\geq 0$, $m\geq 1$, $2g-2+m>0$, are computed recursively as
\begin{align} \label{eq:CEO-TR-original}
	\omega^{(g)}_{m}(z_{\llbracket m \rrbracket}) & = \frac 12 \sum_{i=1}^N \res_{z=p_i} \frac{\int_z^{\sigma_i (z)} \omega^{(0)}_{2}(z_1,\cdot) }{\omega^{(0)}_{1}({\sigma_i (z)}) - \omega^{(0)}_{1}(z)} \Bigg( \omega^{(g-1)}_{m+1} (z,\sigma_i(z),z_{\llbracket m \rrbracket \setminus \{1\}})
	\\ \notag & \qquad + \sum_{\substack{g_1+g_2=g \\ I_1\sqcup I_2 = \llbracket m \rrbracket \setminus \{1\}\\ (g_i,|I_i|)\not= (0,0)}} \omega^{(g_1)}_{|I_1|+1}(z,z_{I_1})\omega^{(g_2)}_{|I_2|+1}(\sigma_i(z),z_{I_2}) \Bigg).
\end{align}

A similar formula can be applied in the case of LogTR. Recall that $a_1,\dots,a_M$ are the simple poles of $dy$ with the residues $\alpha_1^{-1},\dots,\alpha_M^{-1}$, respectively.
We still have ${}^{\log{}}\omega^{(0)}_1 \coloneqq -ydx$ and ${}^{\log{}}\omega^{(0)}_2\coloneqq B$, and then for $g\geq 0$, $m\geq 1$, $2g-2+m>0$, we compute ${}^{\log{}}\omega^{(g)}_{m}$ recursively as
\begin{align} \label{eq:CEO-TR-log}
	{}^{\log{}}\omega^{(g)}_{m}(z_{\llbracket m \rrbracket}) & = \frac 12 \sum_{i=1}^N \res_{z=p_i} \frac{\int_z^{\sigma_i (z)} {}^{\log{}}\omega^{(0)}_{2}(z_1,\cdot) }{{}^{\log{}}\omega^{(0)}_{1}({\sigma_i (z)}) - {}^{\log{}}\omega^{(0)}_{1}(z)} \Bigg( {}^{\log{}}\omega^{(g-1)}_{m+1} (z,\sigma_i(z),z_{\llbracket m \rrbracket \setminus \{1\}})
	\\ \notag & \qquad + \sum_{\substack{g_1+g_2=g \\ I_1\sqcup I_2 = \llbracket m \rrbracket \setminus \{1\}\\ (g_i,|I_i|)\not= (0,0)}} {}^{\log{}}\omega^{(g_1)}_{|I_1|+1}(z,z_{I_1}){}^{\log{}}\omega^{(g_2)}_{|I_2|+1}(\sigma_i(z),z_{I_2}) \Bigg)
	\\ \notag
	& \quad - \delta_{m,1} \sum_{i=1}^M \res_{z=a_i} \int_{a_i}^{z} {}^{\log{}}\omega^{(0)}_{2}(z_1,\cdot) [\hbar^{2g}] \left( \frac{1}{\alpha_i\cS(\alpha_i \hbar \partial_x)} \log(z-a_i)\right)dx.
\end{align}

\begin{remark} The recursion procedures express $\omega^{(g)}_m$ and ${}^{\log{}}\omega^{(g)}_m$ via $\omega^{(g')}_{m'}$ and, respectively, ${}^{\log{}}\omega^{(g')}_{m'}$ with either $g'<g$ or $g'=g$ and $m'<m$.
\end{remark}

\begin{remark} Note that though Equation~\eqref{eq:CEO-TR-log} differs from Equation~\eqref{eq:CEO-TR-original} only for $m=1$, within the further recursion it affects all terms for any $m$.
\end{remark}

\begin{remark} If the pole of $dy$ is not simple or if $dx$ has also a pole at $a_i$
for $i=1,\dots, M$, then logarithmic topological recursion coincides with the original one.
\end{remark}

\begin{remark} If the two recursions do not coincide, then the first term when  $\omega^{(g)}_m\not={}^{\log{}}\omega^{(g)}_m$ is $(g,m)=(1,1)$. In particular, in genus $0$ we always have $\omega^{(0)}_m={}^{\log{}}\omega^{(0)}_m$, $m\geq 1$.
\end{remark}

\subsection{Recomputation procedure}\label{sec:recomputation}
Adapting the formulas given in the general situation of blobbed topological recursion \cite[Equation~2.7]{BS-blobbed} to the case of LogTR we obtain the following expressions connecting $\{\omega^{(g)}_m\}$ and $\{{}^{\log{}}\omega^{(g)}_m\}$:

\begin{proposition} \label{prop:log-omega-in-ordinary-omega}
For $2g-2+m>0$ we have
	\begin{align} \label{eq:log-omega-in-ordinary-omega}
		{}^{\log{}}\omega^{(g)}_m(z_{\set{m}}) & = \delta_{m,1} \phi_g(z_1)+ \sum_{n=0}^\infty \frac{1}{n!} \sum_{\substack{g_0+g_1+\cdots+g_n = g\\ g_0\geq 0,\,  g_1,\dots,g_n\geq 1}} \omega^{(g_0)}_{m+n}(z_{\set{m}}, t_{\set{n}})\Bigg(\prod_{i=1}^n *\phi_{g_i} (t_i) \Bigg);
		\\
		\label{eq:ordinary-omega-in-log-omega}
		\omega^{(g)}_m(z_{\set{m}}) & = -\delta_{m,1} \phi_g(z_1)+ \sum_{n=0}^\infty \frac{(-1)^n}{n!} \sum_{\substack{g_0+g_1+\cdots+g_n = g\\ g_0\geq 0,\,  g_1,\dots,g_n\geq 1}} {}^{\log{}}\omega^{(g_0)}_{m+n}(z_{\set{m}}, t_{\set{n}})\Bigg(\prod_{i=1}^n *\phi_{g_i} (t_i) \Bigg),
	\end{align}
where
\begin{align}
	\phi_g(t) \coloneqq - \sum_{j=1}^M \res_{z=a_i} \int_{a_i}^{z} B(t,\cdot) [\hbar^{2g}] dx \Big(\frac{1}{\cS(\alpha_i \hbar \partial_x)}-1\Big) \frac{\log(z-a_i)}{\alpha_i},
\end{align}
and for any two differentials $\omega(t)$ and $\phi(t)$ we denote
\begin{align}
	\omega(t)*\phi(t) \coloneqq \sum_{i=1}^N \res_{t=p_i} \omega(t)\cdot \int_{p_i}^t \phi(t').
\end{align}
\end{proposition}

\begin{remark} Equations~\eqref{eq:log-omega-in-ordinary-omega} and \eqref{eq:ordinary-omega-in-log-omega} have a finite number of summands on the right hand side because of the conditions $g_1,\dots,g_n\geq 1$. In particular, $\sum_{n=0}^\infty$ can be replaced by $\sum_{n=0}^g$.
\end{remark}

\begin{remark} Equation~\eqref{eq:log-omega-in-ordinary-omega} also manifestly implies that ${}^{\log{}}\omega^{(0)}_m=\omega^{(0)}_m$ for all $m$.
\end{remark}

\begin{proof}[Proof of Proposition~\ref{prop:log-omega-in-ordinary-omega}] We just have to match the notation of LogTR considered as an instance of blobbed topological recursion to the one used in \cite[Equation~2.7]{BS-blobbed}, and we refer the reader to op. cit. for details.

The projection operator $\cP$ there acts on $1$-forms as $\cP\colon \psi(z) \mapsto \sum_{i=1}^N \res_{t\to p_i} \psi(t) \int^t B(z,t')$, and $\mathcal{H} \psi(z) \coloneqq \psi(z) - \cP \psi(z)$. In the case of LogTR, if $2g-2+m>0$, then
\begin{align}
	\cP_1\cdots\cP_m ({}^{\log{}}\omega^{(g)}_m(z_{\set{m}}) )=
	\begin{cases}
		{}^{\log{}}\omega^{(g)}_m(z_{\set{m}}), & m\geq 2; \\
		{}^{\log{}}\omega^{(g)}_1(z_1) - \phi_g(z_1), & m=1.
	\end{cases}
\end{align}
Thus $\varphi_{g,m}\coloneqq \mathcal{H}_1\cdots\mathcal{H}_m ({}^{\log{}}\omega^{(g)}_m)$ (the ingredients of \cite[Equation~2.7]{BS-blobbed}) are nontrivial only for $m=1$, and in the latter case $\varphi_{g,1}=\phi_g$. With this match of notation in mind, Equation~\eqref{eq:log-omega-in-ordinary-omega} coincides with \cite[Equation~2.7]{BS-blobbed}.
\end{proof}

Note that there is a general relationship between intersections of classes on the moduli spaces of curves and topological recursion~\cite{Eynard-intersections}. Under extra assumptions and in the properly chosen basis of differentials on $\Sigma$ the classes on the moduli spaces of curves form a cohomological field theory~\cite{DOSS}, see also~\cite{DNOPS-1,DNOPS-2}, not necessarily with a flat unit~\cite{LPSZ}. There is a general recipe how to recompute the classes on the moduli spaces of curves in the presence of blobs; it is described in~\cite[Section 3]{BS-blobbed} and it can be easily adapted to LogTR / original TR correspondence described in Proposition~\ref{prop:log-omega-in-ordinary-omega}. We remark here that in the special case of $dx$ having one zero, the recipe in~\cite[Section 3, Theorem 3.9]{BS-blobbed} reduces to an insertion of an extra exponential factor of $\kappa$-classes, and we indeed see this effect in the example described in Section~\ref{sec:kappaInsertion}.

\section{Proof of the main theorem}\label{sec:proofs}

The proof of Theorem~\ref{thm:LogXYSwap} follows similar lines to the proofs of~\cite[Theorems 1.8 and 1.14]{alexandrov2022universal}. In order to prove it, we have to analyze the local behavior of the right hand side of Equations~\eqref{eq:MainFormulaSimple} and~\eqref{eq:MainFormulaInverse} near the special points, that is, near the zeros of $dx$ and $dx^\vee$, near the diagonals $z_i=z_j$, and near the points of logarithmic singularities of $y$ and $y^\vee$. Note that Equations~\eqref{eq:MainFormulaSimple} and~\eqref{eq:MainFormulaInverse} imply that there are no other points where $\{{}^{\log{}}\omega^{(g)}_n\}$ might have singularities given that $\{{}^{\log{}}\omega^{\vee,(g)}_n\}$ satisfy the logarithmic projection property and vice versa, so once we complete the analysis at these points, we complete the proof.

We assume in the proof that $x$ and $y$ satisfy assumptions of Theorem~\ref{thm:LogXYSwap}, that is, $dx$ and $dy$ are meromorphic, have simple zeros only, $dy$ is regular and non-vanishing at zeros of $dx$, an $dx$ is regular and non-vanishing at zeros of~$dy$. Note that in Lemmas \ref{lem:expectedpoles} and \ref{lem:nopoles} we do not make these assumptions because these lemmas work in a more general setup.

In this section we always assume that $\{{}^{\log{}}\omega^{\vee,(g)}_n\}$ are \textit{defined} via formula \eqref{eq:MainFormulaSimple} through $\{{}^{\log{}}\omega^{(g)}_n\}$. Starting from Lemma \ref{lem:hodypoles} we will always assume that $\{{}^{\log{}}\omega^{(g)}_n\}$  satisfy LogTR on the spectral curve $(\Sigma,x,y)$.


\subsection{Loop equations and poles}\label{sec:loopeqpole}


 Note that \cite[Propositions 4.6 and 5.7]{alexandrov2022universal} hold for our LogTR case in a completely analogous way, since the projection property was not used in those propositions. We have the following corollary:

\begin{lemma} \label{lem:expectedpoles} If the differentials  $\{{}^{\log{}}\omega^{(g)}_n\}$ are regular on the diagonals for $(g,n)\ne(0,2)$ then the differentials $\{{}^{\log{}}\omega^{\vee,(g)}_n\}$ defined by~\eqref{eq:MainFormulaSimple} are also regular on the diagonals for $(g,n)\ne(0,2)$, the right hand side of~\eqref{eq:MainFormulaInverse} is well defined and~\eqref{eq:MainFormulaInverse} holds true.

Furthermore, assume that the zeros of $dx$ and $dx^\vee$ are simple and disjoint and $\{{}^{\log{}}\omega^{(g)}_n\}$ and $\{{}^{\log{}}\omega^{\vee,(g)}_n\}$ are related by (mutually inverse) Equations~\eqref{eq:MainFormulaSimple} and~\eqref{eq:MainFormulaInverse}. Then
$\{{}^{\log{}}\omega^{(g)}_n\}$ (respectively, $\{{}^{\log{}}\omega^{\vee,(g)}_n\}$) satisfy the linear and quadratic loop equations at the zeros of $dx$ (respectively, $dx^\vee$) if and only if $\{{}^{\log{}}\omega^{\vee,(g)}_n\}$ (respectively, $\{{}^{\log{}}\omega^{(g)}_n\}$) are regular at the zeros of $dx$ (respectively, $dx^\vee$).
\end{lemma}

This proves that once $\{{}^{\log{}}\omega^{(g)}_n\}$ satisfy LogTR and $\{{}^{\log{}}\omega^{\vee,(g)}_n\}$ are given by Equation~\eqref{eq:MainFormulaSimple}, then $\{{}^{\log{}}\omega^{\vee,(g)}_n\}$ satisfy the blobbed topological recursion, and vice versa. It remains to check the logarithmic projection property. To this end, we have to prove that $\{{}^{\log{}}\omega^{\vee,(g)}_n\}$ given by Equation~\eqref{eq:MainFormulaSimple} are not singular at the poles of $dy$ and that they have the prescribed singular behavior at the poles  of $dy^\vee$. This is done in the next two subsections.

\subsection{Elimination of logarithmic singularities}\label{sec:elim_log}


Let us show that the dual differentials $\{{}^{\log{}}\omega^{\vee,(g)}_n\}$ are regular at the poles of~$dy$. We start with a lemma related to the case of a LogTR-vital pole, where the suitable $\hbar$-deformation of the function~$y$ saves the day.

\begin{lemma}\label{lem:nopoles}
Let $z$ be a local coordinate near a point of the spectral curve (s.t. $z=0$ at this point). Assume that $dy$ has a simple pole at this point so that
\begin{equation}
y=\frac{1}{q}\log(z)+(\text{holomorphic)},\quad
\partial_y=(q\,z+o(z))\partial_z
\end{equation}
for some constant $q\ne0$. Assume also that the vector field $\partial_x$ is holomorphic and non-vanishing at this point. Then, for an arbitrary function $f$ polynomial in $u$ and holomorphic in~$z$, the following $1$-form
\begin{equation} \label{eq:omegaveeonedypolesholo}
dy \sum_{r=0}^\infty\partial_y^r\;[u^r]\frac{dx}{dy}e^{-u\bigl(\frac{\cS(u\hbar\partial_x)}{\cS(q\hbar\partial_x)}-1\bigr)\frac{\log(z)}{q}}f(u,z)
\end{equation}
is holomorphic at $z=0$.
\end{lemma}

\begin{proof}

Since we have assumed that $\partial_x$ is holomorphic and non-vanishing at the origin,
$x$ itself can be taken as a local coordinate. So we can assume that $x=z$ without loss of generality. The case of a general nonzero value $q$ can be reduced to the case $q=1$ by a simple rescaling arguments. So it is sufficient to consider the case $q=1$ without loss of generality. Consider the coefficient of any power of $\hbar$ in the following function
\begin{equation}
z e^{-u\bigl(\frac{\cS(u\hbar\partial_z)}{\cS(\hbar\partial_z)}-1\bigr)\log(z)}.
\end{equation}
It is a Laurent series in~$z$, and the coefficient of any power in~$z$ is a polynomial in~$u$ and a monomial in $\hbar$. Moreover, it is proved in~\cite[Lemma 4.1 and 4.2]{bychkov2021topological} that the coefficient of $z^{-k}$ for $k\ge0$ is divisible by $u(u+1)\dots(u+k)$.  This property of the Laurent monomials entering the series is obviously preserved when we multiply this series by a holomorphic function~$z^{-1}\frac{dz}{dy}f(z,u)$.

Next, we need to replace any appearance of $u$ in this coefficient by the operator $\partial_y=(z+o(z))\partial_z$. Let us do it just once, replacing $u+k$ by $\partial_y+k$ in the last factor. Then we obtain a linear combination of similar monomials but with smaller exponents~$k$. Proceeding by induction, we obtain that all non-positive powers of $z$ cancel out. It follows that any such monomial is annihilating by the operator $\sum_{r=0}^\infty \partial_y^r[u^r]\cdot$. So, the result of application of this operator is a holomorphic function in~$z$ vanishing at $z=0$. Multiplying the result by $dy$ having a simple pole we get a $1$-form which is still holomorphic.
\end{proof}

Now let us discuss the case of the poles of $dy$ of higher order.
\begin{lemma}\label{lem:hodypoles}
	Let $x$ and $y$ satisfy the assumptions of Theorem~\ref{thm:LogXYSwap}.
	
	Let $z$ be a local coordinate near a point of the spectral curve (s.t. $z=0$ at this point). Assume that $dy$ has a pole of order greater than one at this point, so that
	\begin{equation}
		dy=\left(\dfrac{1}{\alpha z^{m}} + o\left(\dfrac{1}{z^{m}}\right)\right)dz,\quad 
		\partial_y=(\alpha\,z^m+o(z^m))\partial_z
	\end{equation}
	for $m\in \mathbb{Z}_{\geq 2}$ and some constant $\alpha\ne0$. Assume also that $dx$ does not have a pole at $z=0$.
	Then the right hand side of Equation~\eqref{eq:MainFormulaSimple} is holomorphic at $z=0$.
\end{lemma}
\begin{proof}
	We follow~\cite[Section~5.3]{alexandrov2022universal}. 
	
 For a while let us forget about the assumption that $dx$ is regular at $z = 0$ and consider
	\begin{equation}
		dx= \left(\dfrac{1}{\beta z^s}+o\left(\dfrac{1}{z^s}\right)\right)dz,
	\end{equation}
	where $s\in\mathbb{Z}_{\geq 0}$  ($s$ cannot be negative since under our assumptions poles of $dy$ cannot coincide with zeros of $dx$) 
	and $\beta\neq 0$ is some constant. We waive the assumption in order to refer to the following more general computation in the subsequent lemmas.
 	
	Note that since $z=0$ is not a LogTR-vital singularity of $y$, the only way a pole at this point enters~\eqref{eq:MainFormulaSimple} is through ${}^{\log{}}\omega^{(0)}_1$ in the second line and through the factor $dx/dy$ if $z=0$ is simultaneously a pole of $dx$; moreover ${}^{\log{}}\omega^{(0)}_1$ is being acted upon by $\hbar u \partial_x$ at least twice.
	

	
	Consider a term in the expansion of the exponential in the second line of~\eqref{eq:MainFormulaSimple} coming from the product of the $r_1$-th, \dots, $r_\ell$-th terms in the expansion of $\mathcal{S}$. For the $r_i$-th term $expr_{r_i}$ coming from the expansion of $\mathcal{S}$ in the exponent we have
	\begin{equation}
		expr_{r_i} \sim u^{2r_i+1} (\partial_x)^{2r_i} \dfrac{1}{z^{m-1}} \sim u^{2r_i+1} z^{-m+1+2r_i(s-1)}.
	\end{equation}
	For the product of these terms in the expansion of the exponential (which is also divided by $u$) we have
	\begin{align}
		1/u \prod_{i=1}^\ell expr_{r_i} \sim u^{\ell-1+ 2\sum r_i } z^{-\ell(m-1)+2(s-1)\sum r_i}.
	\end{align}
	For the factor $dx/dy$ we have $dx/dy\sim z^{m-s}$. After that note that $u$ gets replaced by $\partial_y\sim z^m \partial_z$, which means that the whole expression is
	\begin{align}\label{eq:ypoles}
		\sim z^{(m-1)(\ell-1+ 2\sum r_i )} z^{-\ell(m-1)+2(s-1)\sum r_i} z^{m-s} = z^{1-s+2(m+s-2)\sum r_i}.
	\end{align}
	Since $m\geq 2$ and $s\geq 0$ this is regular for $\sum r_i\geq 1$.
	
	If we recall our assumption that $dx$ does not have a pole at $z=0$ then we are done, as $s=0$ in this case. However, keep the full formula~\eqref{eq:ypoles} in mind, as we will use it later for analyzing the poles of $dx$.
\end{proof}

Having the previous two lemmas and the result from Section~\ref{sec:loopeqpole}, we can state the following
\begin{lemma}\label{lem:loopandpoles}
Let $x$ and $y$ satisfy the assumptions of Theorem~\ref{thm:LogXYSwap}, and let $\{{}^{\log{}}\omega^{(g)}_n\}$	satisfy the LogTR on the respective spectral curve with ${}^{\log{}}\omega^{(0)}_1 \coloneqq -ydx$ and ${}^{\log{}}\omega^{(0)}_2\coloneqq B$.

Then $\{{}^{\log{}}\omega^{\vee,(g)}_n\}$ given by formula~\eqref{eq:MainFormulaSimple} are (for $2g-2+n>0$) regular at all poles of $dy$ which are not simultaneously poles of $dx$. 
\end{lemma}
\begin{proof}
	
	Lemma~\ref{lem:nopoles} implies that there are no poles at the simple poles of $dy$ (at those of them that do not coincide with the poles of $dx$). Indeed, by Definition~\ref{def:log-proj} the differentials $\{{}^{\log{}}\omega^{(g)}_n\}$ can only have poles at the poles of $dy$ for $n=1$, and thus these poles can only appear in the second line of~\eqref{eq:MainFormulaSimple} (and not in the third line, as the summation goes over graphs with multiedges of index at least two); and moreover, the principal parts of $\{{}^{\log{}}\omega^{(g)}_1\}$ at these points are prescribed by expressions~\eqref{eq:PrincipalPartsLog}. Substituting~\eqref{eq:PrincipalPartsLog} into the second line of~\eqref{eq:MainFormulaSimple} and taking into account that ${}^{\log{}}\omega^{(0)}_1 \coloneqq -ydx$, one gets an expression precisely of the form~\eqref{eq:omegaveeonedypolesholo}, which is regular by Lemma~\ref{lem:nopoles}.
	
	Lemma~\ref{lem:hodypoles} states that the differentials $\{{}^{\log{}}\omega^{\vee,(g)}_n\}$ are regular at the poles of $dy$ of higher order (at those of them that do not coincide with the poles of $dx$).
	
	
\end{proof}



What remains to show to complete the proof of LogTR for $\{{}^{\log{}}\omega^{\vee,(g)}_n\}$ is to analyze the poles of $dy^\vee=dx$ and to show that that the differentials $\{{}^{\log{}}\omega^{\vee,(g)}_n\}$ have the prescribed singular behavior at the LogTR-vital poles of $dy^\vee$ and are regular at all the other ones; that is, that they satisfy the logarithmic projection property.

\subsection{Emergence of the logarithmic projection property}


First we consider the most essential case of a simple pole of $dx$.

\begin{lemma}\label{lem:logemerg} Let $x$ and $y$ satisfy the assumptions of Theorem~\ref{thm:LogXYSwap}.
	
	Let $z$ be a local coordinate near a point of the spectral curve (s.t. $z=0$ at this point). Assume that $dx$ has a simple pole at this point so that
	\begin{equation}
		x=\frac{1}{q}\log(z)+\text{(holomorphic)},\quad
		\partial_x=(q\,z+o(z))\partial_z
	\end{equation}
	for some constant $q\ne0$. 
	Then the right hand side of Equation~\eqref{eq:MainFormulaSimple} is holomorphic at $z=0$ for $n\not=1$ and for $n=1$ we have
	\begin{align}\label{eq:omegavee1pole}
		\sum_{g=0}^\infty \hbar^{2g} \dfrac{{}^{\log{}}\omega^{\vee,(g)}_1(z)}{dx^{\vee}} = -\frac{1}{q\cS(q\hbar\partial_{y})} \log z + \text{(holomorphic)}.
	\end{align}	
	If $z=0$ is simultaneously a pole of $dy$, then the right hand side of Equation~\eqref{eq:MainFormulaSimple} is holomorphic at $z=0$ for all $(g,n)\neq(0,1)$.
\end{lemma}

\begin{proof}
	Let us consider where in the formula~\eqref{eq:MainFormulaSimple} a pole at $z=0$ can come from. It comes from the $dx/dy$ factor and also it can come from 
	the ${}^{\log{}}\omega^{0}_1$ term in the second line 
	 if $z=0$ is simultaneously a simple pole of $dy$ (since it is not a LogTR-vital pole of $dy$, the terms with higher $g$ are regular). Note that if the second case is true, that is if $z=0$ is a simple pole of $dy$, then since $\partial_x$ vanishes at the origin in this case, the poles coming from ${}^{\log{}}\omega^{0}_1$ terms in the second line get canceled (since they are logarithmic, and these terms are acted upon by $\partial_x$ at least once). If $z=0$ is a pole of $dy$ of higher order, we can use formula~\eqref{eq:ypoles} for the contribution of the exponential from the second line of~\eqref{eq:MainFormulaSimple} together with $dx/dy$. Note that it is clear from that formula that the exponential only gives a holomorphic factor and we are only interested in the $dx/dy$ factor.
	
	
	
	The last summand on the right hand side of~\eqref{eq:MainFormulaSimple} gives $-\frac{1}{q}\log(z)+ \text{\emph{(holomorphic)}}$, which agrees with the $\hbar=0$ part of the right hand side of~\eqref{eq:omegavee1pole} 
	. As for the first summand of~\eqref{eq:MainFormulaSimple}, it does not give any contribution to the $(g,n)=(0,1)$ case.
	
Let us study the first summand. Note that the $dx/dy$ factor contributes a simple pole, and it cancels with a zero that would be contributed by any non-trivial term in the expansion of the exponent in the second line and/or any term that would come with any multi edge in the third line other than the self-adjusted $(0,2)$-multiedges. Thus, ${}^{\log{}}\omega^{\vee,(g)}_n$ is regular for $n>1$, and in the case $n=1$, $g\geq 1$, we have
	\begin{align}
		{}^{\log{}}\omega^{\vee,(g)}_1(z) & =
		-
		\coeff \hbar {2g} dy
		\sum_{k=0}^\infty  \partial_{y}^{k} [u^{k}] \frac{dx}{dy}
\frac{1}{u} e^{\frac12
		\restr{z_1}{z} \restr{z_2}{z} \hbar^2 u^2\cS(u\hbar\partial_{x_1})\cS(u\hbar\partial_{x_2})
\left(\frac{\frac{dz_1}{dx_1}\frac{dz_2}{dx_2}}{(z_1-z_2)^2}-\frac{1}{(x_1-x_2)^2}\right)}
		\\ \notag & \quad
		+ \text{(holomorphic)}.
	\end{align}	
	Here we took only the contributions solely containing $(0,2)$-multiedges; the exponential comes from the sum over all such graphs (with a single  vertex and any number of edges of index 2), with the $1/|\Aut(\Gamma)|$ factor providing the necessary factorial.
	
	The local computation performed in the proof of Proposition~\ref{prop:DualTRTrivial} (cf.~Equation~\eqref{eq:LoopContributions}) implies that denoting $w=e^{\frac{u\hbar}{2}\partial_x}z=e^{\frac{q u\hbar}{2}}z+O(z^2)$, $\bar w=e^{-\frac{u\hbar}{2}\partial_x}z=e^{-\frac{q u\hbar}{2}}z+O(z^2)$, we can rewrite the latter expression as
	\begin{align} \label{eq:omega1dxsingcomp}
		\dfrac{{}^{\log{}}\omega^{\vee,(g)}_1(z)}{dy} & =
		-
		\coeff \hbar {2g}
		\sum_{k=0}^\infty  \partial_{y}^{k} [u^{k}] \frac{dx}{dy}
		\frac{1}{u}  \frac{\hbar u\sqrt{\frac{dw}{dx}\frac{d\bar w}{dx}}}{w-\bar w}
		+ \text{(holomorphic)}
		\\ \notag & =
		-
		\coeff \hbar {2g}
		\sum_{k=0}^\infty  \partial_{y}^{k} [u^{k}] \frac{dz}{dy}
		\frac{\hbar }{e^{\frac{qu\hbar}{2}}-e^{-\frac{qu\hbar}{2}}}\frac{1}{z}
		+ \text{(holomorphic)}
		\\ \notag & =
		-
		\coeff \hbar {2g}
		\sum_{k=0}^\infty  \partial_{y}^{k+1} [u^{k}]
		\frac{1}{q u \cS(q u \hbar)}\log z
		+ \text{(holomorphic)}
		\\ \notag & =
		-
		\coeff \hbar {2g}
		\frac{1}{q\cS(q\hbar \partial_y)}\log z
		+ \text{(holomorphic)},
	\end{align}	
	which establishes the right singular behavior at $z=0$ prescribed by the logarithmic projection property. The last equality in~\eqref{eq:omega1dxsingcomp} is only valid since $g\geq 1$; in this case it does not matter that the sum in the second-to-last line does not have a term with $\partial_y^0$ while in the last line it is present (taking $\coeff \hbar {2g}$ kills it).

	Note that the $g>0$ parts (i.e. the coefficients in front of $\hbar^{2g}$ for $g>0$ in the $\hbar$-series expansion) of the RHS of~\eqref{eq:omegavee1pole} are regular if $z=0$ is not a LogTR-vital singularity of $dx$ (i.e. when it is simultaneously a pole of $dy$ in this case), since $\partial_y$ kills the pole by acting on it at least once, thus the final statement of the lemma holds.
\end{proof}

\begin{lemma}\label{lem:hodxpoles}Let $x$ and $y$ satisfy the assumptions of Theorem~\ref{thm:LogXYSwap}.
	
	Let $z$ be a local coordinate near a point of the spectral curve (s.t. $z=0$ at this point). Assume that $dx$ has a pole of order greater than one at this point, so that
	\begin{equation}
		dx=\left(\dfrac{1}{\alpha z^{m}} + o\left(\dfrac{1}{z^{m}}\right)\right)dz,\quad 
		\partial_x=(\alpha\,z^m+o(z^m))\partial_z
	\end{equation}
	for $m\in \mathbb{Z}_{\geq 2}$ and some constant $\alpha\ne0$.
	Then the right hand side of Equation~\eqref{eq:MainFormulaSimple} is holomorphic at $z=0$.
\end{lemma}
\begin{proof}
	The proof goes similar to the proof of the previous Lemma~\ref{lem:logemerg}. In fact, exactly the same reasoning works up until the first line of~\eqref{eq:omega1dxsingcomp} inclusively.
	
	Indeed, 
	by exactly the same reasoning we are only interested in the pole coming from the $dx/dy$ factor even if $dy$ also has a pole at $z=0$.
	
	
	Then we arrive exactly at the first line of~\eqref{eq:omega1dxsingcomp}, but we have different series expansions for $w$ and $\bar w$ in this case:
	\begin{align}
		w&=e^{\frac{u\hbar}{2}\partial_x}z=z+\dfrac{u\hbar}{2}\partial_x z+\dfrac{u^2\hbar^2}{8}\partial^2_x z + O(u^3z^{3m-2}),
		\\
		\bar w&=e^{-\frac{u\hbar}{2}\partial_x}z=z-\dfrac{u\hbar}{2}\partial_x z+\dfrac{u^2\hbar^2}{8}\partial^2_x z + O(u^3z^{3m-2}).
	\end{align}
	Note that now
	\begin{align} \label{eq:wbmw}
		w-\bar w &= u\hbar \partial_x z + O(u^3z^{3m-2}), \\ \label{eq:dwdwb}
		dw d\bar w & = (dz)^2 \left(1+\dfrac{u^2\hbar^2}{4}\left(\partial_z\partial_x^2z -\left(\partial_z\partial_x z\right)^2\right)+o(u^2z^{2m-2})\right).
	\end{align}
	Thus in this case
	\begin{align} 
		\dfrac{{}^{\log{}}\omega^{\vee,(g)}_1(z)}{dy} & =
		-
		\coeff \hbar {2g}
		\sum_{k=0}^\infty  \partial_{y}^{k} [u^{k}] \frac{dx}{dy}
		\frac{1}{u}  \frac{\hbar u\sqrt{\frac{dw}{dx}\frac{d\bar w}{dx}}}{w-\bar w}
		+ \text{(holomorphic)}
		\\ \notag & =
		-
		\coeff \hbar {2g}
		\sum_{k=0}^\infty  \partial_{y}^{k} [u^{k}] \frac{dz}{dy}
		\,\frac{1 }{ u }\,\frac{1+O(z^{2m-2})}{\partial_xz+O(u^2z^{3m-2})}
		+ \text{(holomorphic)}.
	\end{align}	
	The first equality is precisely the same as in the first line of~\eqref{eq:omega1dxsingcomp}, and the second line is obtained by substituting the expressions~\eqref{eq:wbmw} and~\eqref{eq:dwdwb}. Taking into account that $\partial_xz= \alpha z^m +O(z^{m+1})$, we can rewrite this as
	\begin{align}
		&\dfrac{{}^{\log{}}\omega^{\vee,(g)}_1(z)}{dy}\\\notag
		&\phantom{==}=-
		\coeff \hbar {2g}
		\sum_{k=0}^\infty  \partial_{y}^{k} [u^{k}] \frac{dz}{dy}
		\,\frac{1 }{ u }\,\frac{1}{\partial_x z}\left(1+O(z^{2m-2})\right)\left(1+O(z^{3m-2}/z^{m})\right)
		+ \text{(holomorphic)}
	\end{align}
	Taking into account that $dz/dy$ and $\partial_x z$ do not depend on $u$ and the sum over $k$ goes from zero, thus the $1/u$ term disappears, we have
	\begin{align}
		&\dfrac{{}^{\log{}}\omega^{\vee,(g)}_1(z)}{dy}\\\notag
		&\phantom{==}=-
		\coeff \hbar {2g}
		\sum_{k=0}^\infty  \partial_{y}^{k} [u^{k}] \frac{dz}{dy}
		\,\frac{1 }{ u }\,\frac{O(z^{2m-2})}{\partial_x z}
		+ \text{(holomorphic)} \\ \notag
		&\phantom{==}= \text{(holomorphic)}.
	\end{align}
	The last equality holds due to the fact that $\partial_xz= \alpha z^m +O(z^{m+1})$ and $m\geq 2$. 
\end{proof}


\subsection{Proof of the logarithmic \texorpdfstring{$x$--$y$}{x-y} swap theorem}
At this point we have everything we need to prove the main Theorem~\ref{thm:LogXYSwap}.
\begin{proof}[Proof of Theorem~\ref{thm:LogXYSwap}]
	Lemma~\ref{lem:expectedpoles} implies that the expressions $\{{}^{\log{}}\omega^{\vee,(g)}_n\}$ defined through~\eqref{eq:MainFormulaSimple} satisfy the loop equations. Moreover, it states that they have no diagonal poles (apart from the $(g,n)=(0,2)$ case).
	
	Combining Lemma~\ref{lem:loopandpoles} and Lemmas~\ref{lem:logemerg} and~\ref{lem:hodxpoles},
	we conclude that for $2g-2+n>0$ the differentials $\{{}^{\log{}}\omega^{\vee,(g)}_n\}$ do not have any poles apart from the poles at the zeros of $dx^\vee$ and (for $n=1$) the poles at LogTR-vital singular points of $dy^\vee$, where they have behavior described in~\eqref{eq:omegavee1pole}. This precisely corresponds to the logarithmic projection property being satisfied. 
	
	This completes the proof of Theorem~\ref{thm:LogXYSwap}.
\end{proof}

\printbibliography
\end{document}